\documentclass[11pt]{article}
\usepackage[protrusion=true,expansion=true]{microtype}

\usepackage{amsmath,amssymb,amsfonts,amsthm}
\usepackage{subcaption}
\usepackage{graphicx}
\usepackage{algorithm}
\usepackage{breakcites}
\usepackage[noend]{algpseudocode}
\usepackage{fullpage}
\usepackage[backref=page]{hyperref}
\usepackage{color}
\usepackage{wrapfig}
\usepackage{tikz}
\usetikzlibrary{decorations.pathreplacing}
\usepackage{setspace}
\usepackage[framemethod=tikz]{mdframed}
\usepackage{xspace}
\usepackage{pgfplots}
\usepackage{framed}
\usepackage{thmtools} 
\usepackage{thm-restate}
\usepackage{tabu}
\pgfplotsset{compat=1.5}

\newtheorem{theorem}{Theorem}[section]
\newtheorem{corollary}[theorem]{Corollary}
\newtheorem{lemma}[theorem]{Lemma}

\newtheorem{definition}[theorem]{Definition}

\newenvironment{proofof}[1]{\begin{trivlist} \item {\bf Proof
#1:~~}}
  {\qed\end{trivlist}}
\renewenvironment{proofof}[1]{\par\medskip\noindent{\bf Proof of #1: \ }}{\hfill$\Box$\par\medskip}
\newcommand{\namedref}[2]{\hyperref[#2]{#1~\ref*{#2}}}
\newcommand{\thmlab}[1]{\label{thm:#1}}
\newcommand{\thmref}[1]{\namedref{Theorem}{thm:#1}}
\newcommand{\lemlab}[1]{\label{lem:#1}}
\newcommand{\lemref}[1]{\namedref{Lemma}{lem:#1}}

\newcommand{\corlab}[1]{\label{cor:#1}}
\newcommand{\corref}[1]{\namedref{Corollary}{cor:#1}}
\newcommand{\seclab}[1]{\label{sec:#1}}
\newcommand{\secref}[1]{\namedref{Section}{sec:#1}}
\newcommand{\applab}[1]{\label{app:#1}}
\newcommand{\appref}[1]{\namedref{Appendix}{app:#1}}

\newcommand{\figlab}[1]{\label{fig:#1}}
\newcommand{\figref}[1]{\namedref{Figure}{fig:#1}}
\newcommand{\alglab}[1]{\label{alg:#1}}
\renewcommand{\algref}[1]{\namedref{Algorithm}{alg:#1}}

\def \T    {\mdef{\mathcal{T}}}
\def \R    {\mdef{\mathcal{R}}}
\def \A    {\mdef{\mathcal{A}}}
\def \cL    {\mdef{\mathcal{L}}}
\def \bA    {\mdef{\mathbf{A}}}
\def \bM    {\mdef{\mathbf{M}}}
\def \a    {\mdef{\mathbf{a}}}
\def \b    {\mdef{\mathbf{b}}}
\def \r    {\mdef{\mathbf{r}}}
\def \x    {\mdef{\mathbf{x}}}
\def \y    {\mdef{\mathbf{y}}}
\def \z    {\mdef{\mathbf{z}}}
\def \S    {\mdef{\mathcal{S}}}
\def \ams    {\mdef{\textsc{FreqEst}}}
\def \estimate    {\mdef{\textsc{Estimate}}}
\def \countsketch    {\mdef{\textsc{CountSketch}}}
\def \counter    {\mdef{\textsc{Counter}}}

\newcommand\norm[1]{\left\lVert#1\right\rVert}
\newcommand{\PPr}[1]{\ensuremath{\mathbf{Pr}\left[#1\right]}}

\newcommand{\Ex}[1]{\ensuremath{\mathbb{E}\left[#1\right]}}

\renewcommand{\O}[1]{\ensuremath{\mathcal{O}\left(#1\right)}}
\newcommand{\tO}[1]{\ensuremath{\tilde{\mathcal{O}}\left(#1\right)}}
\newcommand{\eps}{\epsilon}

\def \b    {\mdef{\mathfrak{b}}}

\renewcommand{\figurename}{Fig.}

\newcommand{\mdef}[1]{{\ensuremath{#1}}\xspace}  % Math Def which can also be used in normal text.
\DeclareMathOperator*{\polylog}{polylog}
\DeclareMathOperator*{\poly}{poly}
\DeclareMathOperator*{\mc}{mc}
\DeclareMathOperator*{\mmc}{mmc}
\DeclareMathOperator*{\M}{M}
\DeclareMathOperator*{\Span}{Span}

\newcommand{\ignore}[1]{}

\newif\ifnotes\notestrue %set this to true if notes are visible and to false (next line) if they should be hidden
\ifnotes
\newcommand{\samson}[1]{\textcolor{purple}{{\bf (Samson:} {#1}{\bf ) }} \marginpar{\tiny\bf
             \begin{minipage}[t]{0.5in}
               \raggedright S:
            \end{minipage}}}            							
\else
\newcommand{\samson}[1]{}
\fi

\makeatletter
\renewcommand*{\@fnsymbol}[1]{\textcolor{mahogany}{\ensuremath{\ifcase#1\or *\or \dagger\or \ddagger\or
 \mathsection\or \triangledown\or \mathparagraph\or \|\or **\or \dagger\dagger
   \or \ddagger\ddagger \else\@ctrerr\fi}}}
\makeatother

\providecommand{\email}[1]{\href{mailto:#1}{\nolinkurl{#1}\xspace}}

\definecolor{mahogany}{rgb}{0.75, 0.25, 0.0}
\definecolor{darkblue}{rgb}{0.0, 0.0, 0.55}
\definecolor{darkpastelgreen}{rgb}{0.01, 0.75, 0.24}
\definecolor{darkgreen}{rgb}{0.0, 0.2, 0.13}
\definecolor{darkgoldenrod}{rgb}{0.72, 0.53, 0.04}
\definecolor{darkred}{rgb}{0.55, 0.0, 0.0}
\hypersetup{
     colorlinks   = true,
     citecolor    = mahogany,
		 linkcolor		= olive
}
\title{Symmetric Norm Estimation and Regression \\on Sliding Windows}
\author{
Vladimir Braverman\thanks{Johns Hopkins University. 
E-mail: \email{vova@cs.jhu.edu}.}\\
\and
Viska Wei\thanks{Johns Hopkins University. 
E-mail: \email{swei20@jhu.edu}}
\and
Samson Zhou\thanks{Carnegie Mellon University. 
E-mail: \email{samsonzhou@gmail.com}}
}

\begin{document}
\maketitle

\begin{abstract}
The sliding window model generalizes the standard streaming model and often performs better in applications where recent data is more important or more accurate than data that arrived prior to a certain time. We study the problem of approximating symmetric norms (a norm on $\mathbb{R}^n$ that is invariant under sign-flips and coordinate-wise permutations) in the sliding window model, where only the $W$ most recent updates define the underlying frequency vector. Whereas standard norm estimation algorithms for sliding windows rely on the smooth histogram framework of Braverman and Ostrovsky (FOCS 2007), analyzing the \emph{smoothness} of general symmetric norms seems to be a challenging obstacle. Instead, we observe that the symmetric norm streaming algorithm of Braverman \emph{et al.} (STOC 2017) can be reduced to identifying and approximating the frequency of heavy-hitters in a number of substreams. We introduce a heavy-hitter algorithm that gives a $(1+\epsilon)$-approximation to each of the reported frequencies in the sliding window model, thus obtaining the first algorithm for general symmetric norm estimation in the sliding window model. Our algorithm is a universal sketch that simultaneously approximates all symmetric norms in a parametrizable class and also improves upon the smooth histogram framework for estimating $L_p$ norms, for a range of large $p$. Finally, we consider the problem of overconstrained linear regression problem in the case that loss function that is an Orlicz norm, a symmetric norm that can be interpreted as a scale-invariant version of $M$-estimators. We give the first sublinear space algorithms that produce $(1+\eps)$-approximate solutions to the linear regression problem for loss functions that are Orlicz norms in both the streaming and sliding window models. 
\end{abstract}

\section{Introduction}
The efficient estimation of norms is a fundamental problem in the \emph{streaming model}, which  implicitly defines an underlying frequency vector through a series of sequential updates to coordinates of the vector, but each update may only be observed once. 
For example, the $L_2$ and entropy norms are frequently used to detect network anomalies~\cite{KrishnamurthySZC03,ThorupZ04,ChakrabartiBM06}, while the $L_1$ norm is used to monitor network traffic~\cite{FeigenbaumKSV02} and perform low-rank approximation and linear regression~\cite{FeldmanMSW10}, and the top-$k$ and Ky Fan norms are commonly used in matrix optimization problems~\cite{WuDST14}. 
These norms all have the property that they are invariant to permutations and sign flips of the coordinates of the underlying vectors:
\begin{definition}[Symmetric norm]
A norm $\ell:\mathbb{R}^n\to\mathbb{R}$ is a \emph{symmetric norm} if for all $x\in\mathbb{R}^n$ and any $n\times n$ permutation matrix $P$, we have $\ell(x)=\ell(Px)$ and $\ell(x)=\ell(|x|)$, where $|x|$ is the coordinate-wise absolute value of $x$. 
\end{definition}
Symmetric norms include the $L_p$, entropy, top-$k$, $k$-support, and box norms, and many other examples that we detail in \secref{sec:apps}. 
Braverman \emph{et al.}~\cite{BlasiokBCKY17} show that a symmetric norm $\ell$ can be approximated using space roughly $\mmc(\ell)^2$, where $\mmc$ is the \emph{maximum modulus of concentration} of the norm $\ell$, whose formal definition we will defer to \secref{sec:apps}. 
Informally, $\mmc(\ell)$ is roughly the ratio of the maximum value $\ell$ achieves on a unit ball compared to the meidan value of $\ell$ on the unit ball. 

\paragraph{Sliding window model.} 
Unfortunately, the streaming model does not prioritize recent data that is considered more accurate and important than data that arrived prior to a certain time. 
Thus for a number of time-sensitive applications~\cite{BabcockBDMW02,MankuM12,PapapetrouGD15,WeiLLSDW16}, the streaming model has inferior performance compared to the \emph{sliding window model}, in which the underlying dataset consists of only the $W$ most recent updates in the stream. 
The fixed parameter $W>0$ represents the window size for the active data and the goal is to process information about the dataset using space sublinear in $W$. 
Note that the sliding window model is a generalization of the streaming model, e.g., when the stream length $m$ is at most $W$. 
The sliding window model is especially relevant in time-dependent settings such as network monitoring~\cite{CormodeM05a,CormodeG08,Cormode13}, event detection in social media~\cite{OsborneEtAl2014}, data summarization~\cite{ChenNZ16,EpastoLVZ17}, and has been also studied in a number of additional settings~\cite{LeeT06,LeeT06b,BravermanO07,DatarM07,BravermanOZ12,BravermanLLM15,BravermanLLM16,BravermanGLWZ18,BorassiELVZ19,BravermanDMMUWZ20,WoodruffZ20,BorassiELVZ20,JayaramWZ22}.

\paragraph{Problem statement.} 
Formally, the model is as follows. 
Given a symmetric norm $\ell:\mathbb{R}^n\to\mathbb{R}$, we receive updates $u_1,\ldots,u_m$ to the coordinates of an underlying frequency vector $f$. 
Each update with $i\in[m]$ satisfies $u_i\in[n]$ so that the $i$-th update effectively increments the $u_i$-th coordinate of $f$. 
However, in the sliding window model, only the last $W$ updates define $f$ so that for each $j\in[n]$, we have $f_j=|\{i\,:\,u_i=j, i\ge m-W+1\}|$. 
The goal is to approximate $\ell(f)$ at the end of the stream, but $m$ is not given in advance so we cannot simply maintain a sketch of the last $W$ elements because we do not know the value of $m-W+1$ a priori. 

The main challenge of the sliding window model is that updates to $f$ expire implicitly. 
Thus we cannot apply linear sketching techniques, which forms the backbone of many streaming algorithms. 
For example, we do not know that the update $u_{m-W}$ does not affect the value of $f$ until the very last update. 
Thus if we maintain a sketch of the updates that includes $u_{m-W}$, we must ``undo'' the inclusion of $u_{m-W}$ at time $m$; however at that time, it may be too late to remember the value of $u_{m-W}$. 

\subsection{Our Results}
In this paper, we give the first generic framework that can approximate any symmetric norm of an underlying frequency vector in the sliding window model. 
\begin{theorem}
\thmlab{thm:sym:norm}
Given an accuracy parameter $\eps>0$ and a symmetric norm $\ell$, there exists a sliding window algorithm that outputs a $(1+\eps)$-approximation to the $\ell$-norm of the underlying frequency vector with probability $\frac{2}{3}$ and uses space $\mmc(\ell)^2\cdot\poly\left(\frac{1}{\eps},\log n\right)$. 
\end{theorem}
Our framework has specific implications to the well-studied $L_p$ norms and the top-$k$ norm that is used in matrix optimization, as well as the $k$-support, box, and more generally, $Q'$-norms that are frequently used to regularize sparse recovery problems in machine learning. 
We summarize these applications in \figref{fig:results} and provide additional detail on these norms in \secref{sec:apps}. 
\begin{figure*}[!htb]
\begin{center}
{\tabulinesep=1.2mm
\begin{tabu}{|c|c|c|}\hline
Problem & Space Complexity & Reference \\\hline\hline
Symmetric norm $\ell$ & $\mmc(\ell)^2\cdot\poly\left(\frac{1}{\eps},\log n\right)$ & \thmref{thm:sym:norm} \\\hline
$L_p$ norm, $p\in[1,2]$ & $\poly\left(\frac{1}{\eps},\log n\right)$ & \corref{cor:qnorm} \\\hline
$L_p$ norm, $p>2$ & $\poly\left(\frac{1}{\eps},\log n\right)\cdot n^{1-2/p}$ & \corref{cor:lp} \\\hline
$k$-support norm & $\poly\left(\frac{1}{\eps},\log n\right)$ & \corref{cor:qnorm} \\\hline
$Q'$ norm & $\poly\left(\frac{1}{\eps},\log n\right)$ & \corref{cor:qnorm} \\\hline
Box norm & $\poly\left(\frac{1}{\eps},\log n\right)$ & \corref{cor:qnorm} \\\hline
Top-$k$ norm & $\frac{n}{k}\cdot\poly\left(\frac{1}{\eps},\log n\right)$ & \corref{cor:topk} \\\hline
\end{tabu}
}
\end{center}
\vspace{-0.1in}
\caption{Summary of our sliding window algorithms}
\figlab{fig:results}
\end{figure*}
In particular for sufficiently large $p>2$, our $L_p$ norm sliding window algorithm improves upon the $\tO{\frac{1}{\eps^{p+2}}n^{1-2/p}}$ space algorithm by \cite{BravermanO07}.  
Our framework not only uses near-optimal space complexity for these applications, but is also a \emph{universal sketch} that suffices to simultaneously approximate all symmetric norms in a wide parametrizable class.  
\begin{theorem}
\thmlab{thm:sym:universal}
Given an accuracy parameter $\eps>0$ and a space parameter $S$, there exists a sliding window algorithm that uses space $S\cdot\poly\left(\frac{1}{\eps},\log n\right)$ and outputs a $(1+\eps)$-approximation to \emph{any} symmetric norm $\ell$ with $\mmc(\ell)\le\sqrt{S}$, with probability $\frac{2}{3}$.  
\end{theorem}

The general approach to sliding window algorithms is to use the smooth histogram framework by Braverman and Ostrovsky~\cite{BravermanO07}. 
The smooth histogram framework requires the desired objective to be smooth, where given adjacent substreams $A$, $B$, and $C$, a smooth function states that $(1-\eta)f(A\cup B)\le f(B)$ implies $(1-\eps)f(A\cup B\cup C)\le f(B\cup C)$ for some constants $0<\eta\le\eps<1$. 
Intuitively, once a suffix of a data stream becomes a $(1\pm\eta)$-approximation for a smooth function, then it is \emph{always} a $(1\pm\eps)$-approximation, regardless of the subsequent updates that arrive in the stream. 
Since the resulting space complexity depends on $\eta$, this approach requires analyzing the smoothness of each symmetric norm and it is not clear how these parameters relate to $\mmc(\ell)$ or whether there is a general parametrization for each norm.

Instead, we observe that~\cite{BlasiokBCKY17} effectively reduces the problem to computing a $(1+\nu)$-approximation to the frequency of all $\eta$-heavy hitters for a number of various substreams. 
\begin{definition}[$\nu$-approximate $\eta$-heavy hitters]
Given any accuracy parameter $\nu$, a threshold parameter $\eta$, and a frequency vector $f$, an algorithm $\mathcal{A}$ is said to solve the $\nu$-approximate $\eta$-heavy hitters problem if it outputs a set $H$ and a set of approximations $\widehat{f_i}$ for all $i\in H$ such that: 
\begin{enumerate}
\item
If $f_i\ge\eta\norm{f}_2$ for any $i\in[n]$, then $i\in H$. 
That is, $H$ contains all $\eta$-heavy hitters of $f$. 
\item
There exists an absolute constant $C>0$ so that if $f_i\le\frac{C\eta}{2}\norm{f}_2$ for any $i\in[n]$, then $i\notin H$. 
That is, $H$ does not contain any item that is not an $\frac{C\eta}{2}$-heavy hitter of $f$. 
\item
If $i\in H$, then $\mathcal{A}$ reports a value $\widehat{f_i}$ such that $(1-\nu)f_i\le\widehat{f_i}\le(1+\nu)f_i$. 
That is, $\mathcal{A}$ outputs a $(1\pm\nu)$-approximation to the frequency $f_i$, for all $i\in H$. 
\end{enumerate}
\end{definition}
Thus to approximate a symmetric norm on the active elements, it suffices to find $\nu$-approximate $\eta$-heavy hitters for a number of substreams. 
Whereas the sliding window heavy-hitter algorithms~\cite{BravermanGO14,BravermanGLWZ18} optimize for space complexity and only output constant factor approximations to the frequencies of the reported elements, we give a simple modification to their ideas to output $\nu$-approximate $\eta$-heavy hitters. 

\begin{theorem}
\thmlab{thm:hh}
Let $f$ be a frequency vector on $[n]$ induced by the active window of an insertion-only data stream. 
For any accuracy parameter $\nu\in\left(0,\frac{1}{4}\right)$ and threshold $\eta\in(0,1)$, there exists a one-pass streaming algorithm that outputs a list that includes all $\eta$-heavy hitters and no element that is not a $\frac{\eta}{8}$-heavy hitter. 
Moreover, the algorithm reports a $(1+\nu)$-approximation to the frequency $f_i$ of all reported items $i$. 
The algorithm uses $\O{\frac{1}{\nu^3\eta^2}\log^3 n}$ bits of space and succeeds with high probability. 
\end{theorem}

In summary, our main conceptual contribution is the existence of a $(1+\eps)$-approximation algorithm for general symmetric norms in the sliding window model. 
Our technical contributions include an overall framework that incorporates \emph{any} symmetric norm in a plug-and-play manner as well as a heavy-hitter subroutine that may be of independent interest. 
Finally, we perform a number of empirical evaluations comparing our algorithms to uniform sampling on large-scale real-world datasets. 

\paragraph{Independent and concurrent related work.} 
Independent of our work, \cite{KrauthgamerR19} has given a framework for subadditive functions that extends beyond the smooth histogram approach of \cite{BravermanO07}. 
In particular, their framework gives a $(2+\eps)$-approximation for symmetric norms in the sliding window model.  
By comparison, our algorithm achieves a $(1+\eps)$-approximation for symmetric norms on sliding windows. 
Their techniques are based on black-boxing the streaming algorithm of~\cite{BlasiokBCKY17} that approximates the symmetric norm and initializing various instances of the algorithm as the stream progresses. 
We open up the black box by instead introducing a new heavy-hitter algorithm in the sliding window model and using properties of heavy-hitters and level sets to enable a finer approximation to the symmetric norm, e.g., \cite{IndykW05,BravermanOR15,WoodruffZ18,WoodruffZ21}.  

\paragraph{Symmetric norm regression.} 
As a further application of our work, we consider the fundamental overconstrained linear regression problem in the case that loss function that is a symmetric norm, which includes many standard loss functions such as $L_p$ norms, top-$k$ norms, and $Q'$-norms. 
Specifically, given a data matrix $\bA\in\mathbb{R}^{n\times d}$ and a response vector $\b\in\mathbb{R}^n$ with $n\gg d$, we aim to minimize the optimization problem $\min_{\x\in\mathbb{R}^d}\cL(\bA\x-\b)$, where $\cL:\mathbb{R}^n\to\mathbb{R}$ is a loss function. 
When $\cL$ is a symmetric norm, then the loss function places emphasis on the magnitude of the incorrect coordinates rather than their specific indices. 
In particular, we consider the general case where $\cL$ is an Orlicz norm, which can be interpreted as a scale-invariant version of $M$-estimators. 
Embeddings for $(1+\eps)$-approximate solutions to the linear regression problem for loss functions that are Orlicz norms in the central model, where complete access to $\bA$ is given, was recently studied by \cite{AndoniLSZZ18,SongWYZZ19}.  
We give the first algorithms that produce $(1+\eps)$-approximate solutions to the linear regression problem for loss functions that are Orlicz norms in both the streaming and sliding window models. 
Our algorithms are parametrized by a constant $\Delta$, which represents the aspect ratio of the dataset under the norm. 
\begin{theorem}
Given an accuracy $\eps>0$ and a matrix $\bA\in\mathbb{R}^{W\times d}$ whose rows $\a_1,\ldots,\a_W$ arrive sequentially in a stream $\r_1,\ldots,\r_n$ with condition number at most $\kappa$, there exists both a streaming algorithm and a sliding window algorithm that outputs a $(1+\eps)$ embedding for an Orlicz norm with high probability. 
The algorithms sample $\frac{d^2\Delta}{\eps^2}\log\kappa\polylog n$ rows, with high probability. 
(See \thmref{thm:orlicz:se} and \thmref{thm:orlicz:se:sw}.)
\end{theorem}

\subsection{Preliminaries}
For any positive integer $n$, we use $[n]$ to denote the set $\{1,\ldots,n\}$. 
We say an event occurs with high probability, if the probability of the event occurring is $1-\frac{1}{\poly(n)}$, for any arbitrary polynomial that can be determined from altering constants. 
We use $\polylog(n)$ to suppress polylogarithmic factors. 
For a vector $f\in\mathbb{R}^n$, we use $f_i$ to denote the $i$-th coordinate of $f$. 
We use $\circ$ to denote the vertical concatenation of rows, so that for row vectors $\a_1,\a_2\in\mathbb{R}^d$, we have $\a_1\circ\a_2=\begin{bmatrix}\a_1\\\a_2\end{bmatrix}$. 
The condition number of a matrix $\bA\in\mathbb{R}^{n\times d}$ is the ratio of its largest singular value to its smallest nonzero singular value. 
The condition number of a stream $\a_1\circ\a_2\circ\ldots\circ\a_m$ is the largest condition number of any matrix $\a_i\circ\ldots\circ\a_j$ formed by a consecutive number of rows. 

In the sliding window model, we have a stream of length $m$, where we assume $m=\poly(n)$. 
For each $i\in[m]$, the update $u_i\in[n]$ (if active) corresponds to a single increment to coordinate $u_i$ of the underlying frequency vector of dimension $n$. 
For a window parameter $W>0$, only the most recent $W$ updates define the underlying frequency vector, as the previous updates are expired. 

\begin{definition}[$L_p$ norms]
For a vector $f\in\mathbb{R}^n$ and $p>0$, we have the $L_p$ norm $\norm{f}_p=\left(\sum_{i=1}^n f_i^p\right)^{1/p}$. 
For $p=0$, $L_0$ is not a norm, but nevertheless we define $\norm{f}_0=\left|\{i\in[n]\,:\,f_i\neq 0\}\right|$. 
\end{definition}
\noindent
We require the following streaming and sliding window algorithms. 
\begin{theorem}[CountSketch for heavy-hitters]
\cite{CharikarCF04}
\thmlab{thm:countsketch}
Let $f$ be an underlying frequency vector on $[n]$ implicitly defined on through a dynamic (or insertion-only) stream. 
There exists a one-pass streaming algorithm $\countsketch$ that takes a threshold parameter $\nu>0$ and outputs a list $H$ that contains all indices $i\in[n]$ with $f_i\ge\nu\cdot\norm{f}_2$ and no index $j\in[n]$ with $f_j\le\frac{\nu}{2}\cdot\norm{f}_2$. 
The algorithm uses $\O{\frac{1}{\nu^2}\log^2 n}$ bits of space and succeeds with high probability. 
\end{theorem}

\begin{lemma}[Frequency estimation on sliding windows]
\cite{AlonMS99,BravermanO07}
\lemlab{lem:ams}
Let $C=\frac{17}{16}$. 
There exists a one-pass streaming algorithm $\ams$ that simultaneously outputs $2$-approximations to the $L_2$ norm of the frequency vector induced by \emph{any} suffix of an insertion-only stream. 
That is, for a stream of length $m$, the algorithm maintains a data structure that outputs a value $F$ for any query $W\in[m]$, such that $F\le\norm{f}_2\le C\cdot F$, where $f$ is the frequency vector induced by the last $W$ updates of the stream. 
The algorithm uses $\O{\log^2 n}$ bits of space and succeeds with high probability. 
\end{lemma}

\begin{lemma}[Approximate count of an item]
\lemlab{lem:counter}
\cite{BravermanGLWZ18}
For a stream of length $m$, a time $t\in[m]$, an index $i\in[n]$, and an accuracy parameter $\eta>0$, there exists a deterministic one-pass streaming algorithm $\counter$ that simultaneously outputs a $\left(1+\frac{\eta}{4}\right)$-approximation to the frequency of $i$ between $t$ and all times $u\in[t,m]$. 
The algorithm uses $\O{\frac{1}{\eta}\log^2 n}$ bits of space and never fails. 
\end{lemma}
The smoothness of the $L_2$ norm is instrumental in not only ensuring $\ams$ in \lemref{lem:ams} achieves a constant factor approximation, but also that the space of $\ams$ is polylogarithmic. 
\begin{lemma}[Smoothness of $L_2$ norm]
\lemlab{lem:count}
\cite{BravermanO07}
Let $C=\frac{17}{16}$. 
For an insertion-only stream of length $m$, let $a<b<c\le d\le m$ and $X_a$ be a $C$-approximation to the $L_2$ norm of the frequency vector induced by the updates from time $a$ to time $c$ in the stream (inclusive). 
Let $X_b$ be a $C$-approximation to the $L_2$ norm of the frequency vector induced by the updates from time $b$ to time $c$ in the stream (inclusive). 
Let $Y_a$ be the $L_2$ norm of the frequency vector induced by the updates from time $a$ to $d$ and $Y_b$ be similarly defined from time $b$ to $d$. 
If $X_a\le C\cdot X_b$, then $Y_a\le 2Y_b$. 
\end{lemma}
\noindent

%\begin{lemma}[Smoothness of Counters]
%\lemlab{lem:count}
%For an insertion-only stream of length $m$, let $a<b<c\le d\le m$ and $X_a$ be the frequency of an item $i\in[n]$ from time $a$ to time $c$ in the stream (inclusive). 
%Let $X_b$ be the frequency of $i$ from time $b$ to time $c$ in the stream (inclusive). 
%Similarly, let $Y_a$ be the frequency of $i$ from time $a$ to $d$ and $Y_b$ be the frequency of $i$ from time $b$ to $d$. 
%If $X_b\le X_a\le(1+\eps)X_b$, then $Y_b\le Y_a\le(1+\eps)Y_b$. 
%\end{lemma}
%\begin{proof}
%Let $Z$ be the frequency of $i$ after time $c$ and up to time $d$, so that $Y_a=X_a+Z$ and $Y_b=X_b+Z$. 
%Since the stream is insertion-only, then $Z\ge 0$, so that $X_b\le X_a$ implies $Y_b\le Y_a$. 
%Similarly, $X_a\le(1+\eps)X_b$ implies 
%\[Y_a=X_a+Z\le(1+\eps)X_b+Z\le(1+\eps)(X_b+Z)=(1+\eps)Y_b.\]
%\end{proof}

\section{Linear Regression for Orlicz Norms}
In this section, we describe our algorithm for linear regression for Orlicz norms in the streaming and sliding window models. 
Given a data matrix $\bA\in\mathbb{R}^{n\times d}$ and a response vector $\b\in\mathbb{R}^n$ with $n\gg d$, recall that the goal of the overconstrained linear regression problem is to minimize the optimization problem $\min_{\x\in\mathbb{R}^d}\cL(\bA\x-\b)$, where $\cL:\mathbb{R}^n\to\mathbb{R}$ is a loss function. 
For a function $G$, we define the corresponding Orlicz norm $\|\x\|_G$ of a vector $\x\in\mathbb{R}^n$ to be zero if $\x=0^n$ and to be the unique value $\alpha$ such that $\sum_{i=1}^n G(|x_i|/\alpha)=1$ otherwise if $\x\neq0^n$. 
In order to obtain $(1+\eps)$-approximation to linear regression for Orlicz norms, \cite{SongWYZZ19} makes the assumption that (1) $G$ is a strictly increasing convex function on $[0,\infty)$, (2) $G(0)=0$ and $G(x)=G(-x)$ for all $x\in\mathbb{R}$, and (3) there exists an absolute constant $C_G$ such that for all $0<x<y$, $G(y)/G(x)\le C_G(y/x)^2$. 
We also assume without loss of generality that each coordinate of $\bA$ is an integer that is at most $M$ in magnitude for some large $M=\poly(n)$. 

\begin{definition}[Online $L_1$ sensitivity]
For a matrix $\bA=\a_1\circ\ldots\circ\a_n\in\mathbb{R}^{n\times d}$, let $\bA_i=\a_1\circ\ldots\circ\a_i$ for each $i\in[n]$.  
Then the online $L_1$ sensitivity of a row $\a_i$ is defined as
\[\max_{\x\in\mathbb{R}^{d}}\frac{|\langle\a_i,\x\rangle|}{\|\bA_i\x\|_1}.\]
\end{definition}
The online $L_1$ sensitivities can be efficiently approximated, e.g., see~\cite{CohenEMMP15,CohenMP16,CohenMM17,BravermanDMMUWZ20}. 
Namely, a constant factor approximation to any online $L_p$ sensitivity that is at least $\frac{1}{\poly(n)}$ can be computed in polynomial time using (offline) linear programming. 
Similarly, an additive $\frac{1}{\poly(n)}$ approximation to any online $L_p$ sensitivity that is less than $\frac{1}{\poly(n)}$ can be computed in polynomial time using (offline) linear programming. 

\begin{algorithm}[!htb]
\caption{Subspace embedding for Orlicz norms in the row-arrival streaming model}
\alglab{alg:orlicz}
\begin{algorithmic}[1]
\Require{A stream of rows $\a_1,\ldots,\a_n\in\mathbb{R}^d$, parameter $\Delta>0$, and an accuracy parameter $\eps>0$}
\Ensure{A $(1+\eps)$ subspace embedding for Orlicz norms.}
\State{$\bM\gets\emptyset$}
\State{$\alpha\gets\frac{Cd}{\eps^2}\log n$ with sufficiently large parameter $C>0$}
\For{each row $\a_i$, $i\in[n]$}
\If{$\a_i\in\Span(\bM)$}
\State{$\tau_i\gets2\Delta\cdot\max_{\x\in\mathbb{R}^d,\x\in\Span(\bM)}\frac{|\langle\a_i,\x\rangle|}{\|\bM\x\|_1+|\langle\a_i,\x\rangle|}$}
\Else
\State{$\tau_i\gets 1$}
\EndIf
\State{$p_i\gets\min(1, \alpha\tau_i)$}
\State{With probability $p_i$, $\bM\gets\bM\circ\frac{\a_i}{p_i}$}
\Comment{Online sensitivity sampling}
\EndFor
\State{\Return $\bM$}
\end{algorithmic}
\end{algorithm}

\begin{theorem}[Freedman's inequality]
\cite{Freedman75}
\thmlab{thm:scalar:freedman}
Suppose $Y_0,Y_1,\ldots,Y_n$ is a scalar martingale with difference sequence $X_1,\ldots,X_n$. 
Specifically, we initiate $Y_0=0$ and set $Y_i=Y_{i-1}+X_i$ for all $i\in[n]$. 
Let $R\ge|X_t|$ for all $t\in[n]$ with high probability. 
We define the predictable quadratic variation process of the martingale by $w_k:= \sum_{t=1}^k\underset{t-1}{\mathbb{E}}\left[X_t^2\right]$, for $k\in[n]$. 
Then for all $\eps\ge 0$ and $\sigma^2 > 0$, and every $k \in [n]$, 
\[\PPr{\max_{t \in [k]} |Y_t|>\eps\text{ and } w_k \le \sigma^2}\le 2\exp\left(-\frac{\eps^2/2}{\sigma^2 + R\eps/3} \right).\]
\end{theorem}

\begin{lemma}
\lemlab{lem:subspace}
Let $\mathcal{N}$ be a greedily constructed $\eps$-net such that $\|\bA\x\|_G=1$ for all $\x\in\mathcal{N}$. 
Suppose $\Delta_1$ and $\Delta_2$ are parameters such that $\Delta_1\le|\a_j^\top\x|\le\Delta_2$ for all $\x\in\mathcal{N}$ and $j\in[n]$ and $|\a_j^\top\x|\neq 0$. 
Let $\Delta\ge\frac{G(\Delta_1)\cdot\Delta_2}{G(\Delta_2)\cdot\Delta_1}$ in \algref{alg:orlicz}. 
Then for $\eps>\frac{1}{n}$, \algref{alg:orlicz} outputs a matrix $\bM$ such that for all $\x\in\mathbb{R}^d$,
\[|\|\bM\x\|_G-\|\bA\x\|_G|\le\eps\|\bA\x\|_G,\]
with high probability.
\end{lemma}
\begin{proof}
Let $\x\in\mathbb{R}^d$ be an arbitrary vector with $\|\bA\x\|_G=1$ and suppose $\eps\in(0,1/2)$ with $\eps>\frac{1}{n}$. 
We show via induction that $|\|\bM_j\x\|_G-\|\bA_j\x\|_G|\le\eps\|\bA_j\x\|_G$ for all $j\in[n]$ with high probability, where $\bM_j$ is the reweighted submatrix of the input matrix $\bA$ that has been sampled at time $j$. 
Either $\a_1$ is the zero vector or $p_1=1$ so that $\bM_1=\A_1$ for the base case. 

Suppose the statement holds for all $j\in[n-1]$; we prove it holds for $j=n$. 
We define a martingale $Y_0,Y_1,\ldots,Y_n$ implicitly through the difference sequence $X_1,\ldots,X_n$. 
For $j\ge 1$, we set $X_j=0$ if $Y_{j-1}>\eps\|\bA_{j-1}\x\|_p^p$. 
Otherwise if $Y_{j-1}\le\eps\|\bA_{j-1}\x\|_G$, we set
\begin{equation}
X_j=
\begin{cases}
\left(\frac{1}{p_j}-1\right)G(|\a_j^\top\x|) & \text{ if }\a_j\text{ is sampled in }\bM\\
- G(|\a_j^\top\x|)  & \text{otherwise}.
\end{cases}
\end{equation}
Observe that the sequence $Y_0,\ldots,Y_n$ induced by the differences is indeed a valid martingale because $\Ex{Y_j|Y_1,\ldots,Y_{j-1}}=Y_{j-1}$. 
By definition of the difference sequence, we also have 
\[Y_j=\norm{\bM_j\x}_G-\norm{\A_j\x}_G.\]

If $p_j=1$, then $\a_j$ is sampled in $\bM_j$, so that $X_j=0$. 
Otherwise if $p_j<1$, then 
\[\Ex{X_j^2|Y_1,\ldots,Y_{j-1}}=p_j\left(\frac{1}{p_j}-1\right)^2G(|\a_j^\top\x|)^2+(1-p_j)G(|\a_j^\top\x|)^{2}\le\frac{1}{p_j}G(|\a_j^\top\x|)^2.\]
Moreover, $p_j<1$ implies $p_j=\alpha\tau_j$ so that $\Ex{X_j^2|Y_1,\ldots,Y_{j-1}}\le\frac{1}{\alpha\tau_j}G(|\a_j^\top\x|)^2$. 
By the definition of $\tau_j$ and the inductive hypothesis that $|\|\bM_{j-1}\x\|_G-\|\bA_{j-1}\x\|_G\|\le\eps\|\bA_{j-1}\x\|_G<\frac{1}{2}\|\bA_{j-1}\x\|_G$,
\begin{align*}
\tau_j\ge\frac{2\Delta|\a_j^\top\x|}{\|\bM_{j-1}\x\|_1+|\a_j^\top\x|}=\frac{2|\a_j^\top\x|G(\Delta_1)/\Delta_1}{G(\Delta_2)/\Delta_2(\|\bM_{j-1}\x\|_1+|\a_j^\top\x|)}\ge\frac{2G(|\a_j^\top\x|)}{\|\bM_{j-1}\x\|_G+G(|\a_j^\top\x|)},
\end{align*}
since $G(\Delta_2)/\Delta_2\ge G(|\a_j^\top\x|)\ge G(\Delta_1)/\Delta_1$ for all $\a_j\in\mathbb{R}^d$ and $\x\in\mathbb{R}^d$, given the assumption that $M\ge|\a_j^\top\x|$. 
Thus,
\begin{align*}
\tau_j\ge\frac{G(|\a_j^\top\x|)}{\|\bA_{j-1}\x\|_G+G(|\a_j^\top\x|)}=\frac{G(|\a_j^\top\x|)}{\|\bA_j\x\|_G}\ge\frac{G(|\a_j^\top\x|)}{\|\bA\x\|_G}.
\end{align*}
Consequently, $\sum_{j=1}^n\Ex{X_j^2|Y_1,\ldots,Y_{j-1}}\le\sum_{j=1}^n\frac{\|\bA\x\|_G\cdot G(|\a_j^\top\x|)}{\alpha}\le\frac{\|\A\x\|_G^2}{\alpha}$. 

Moreover, $|X_j|\le\frac{1}{p_j}G(|\a_j^\top\x|)$. 
For $p_j=1$, $\frac{1}{p_j}G(|\a_j^\top\x|)\le\|\bA_j\x\|_G\le\|\bA\x\|_G$. 
On the other hand if $p_j<1$, then $p_j=\alpha\tau_j<1$. 
Again by the definition of $\tau_j$ and by the inductive hypothesis that $|\|\bM_{j-1}\x\|_G-\|\bA_{j-1}\x\|_G\|\le\eps\|\bA_{j-1}\x\|_G<\frac{1}{2}\|\bA_{j-1}\x\|_G$,  
\[\frac{G(|\a_j^\top\x|)}{2\|\A_j\x\|_G}\le\frac{G(|\a_j^\top\x|)}{|\bM_{j-1}\x|_G+G(|\a_j^\top\x|)}\le\tau_j.\] 
Hence for $\alpha=\frac{Cd}{\eps^2}\log n$, 
\[|X_j|\le\frac{1}{p_j}G(|\a_j^\top\x|)\le\frac{2}{\alpha}\|\A_j\x\|_G\le\frac{2\eps^2}{Cd\log n}\|\A_j\x\|_G\le\frac{2\eps^2}{Cd\log n}\|\A\x\|_G.\] 

We apply Freedman's inequality (\thmref{thm:scalar:freedman}) with $\sigma^2=\frac{\|\bA\x\|_G^2}{\alpha}$ for $\alpha=\frac{Cd}{\eps^2}\log n$ and $R\le\frac{2\eps^2}{d\log n}\|\bA\x\|_G$. 
Thus,
\[\PPr{|Y_n|>\eps\|\bA\x\|_G}\le2\exp\left(-\frac{\eps^2\|\bA\x\|_G^2/2}{\|\bA\x\|_G^2/\alpha + \frac{2\eps^2}{d\log n}\|\bA\x\|_G\cdot\eps\|\bA\x\|_G}\right)\le\frac{1}{2^{\Omega(d)}\,\poly(n)},\]
for sufficiently large $C$. 

We now union bound over an $\eps$-net by first defining the unit ball $B=\{\bA\y\in\mathbb{R}^n\,|\,\norm{\bA\y}_G=1\}$. 
We also define $\mathcal{N}$ to be a greedily constructed $\eps$-net of $B$. 
Since balls of radius $\frac{\eps}{2}$ around each point must not overlap while simultaneously all fitting into a ball of radius $r+\frac{\eps}{2}$ for some constant $r>0$, then $\mathcal{N}$ has at most $\left(\frac{3r}{\eps}\right)^d$ points. 
Thus by a union bound for $\frac{1}{\eps}<n$, $|\norm{\bM\y}_G-\norm{\bA\y}_G|\le\eps\norm{\bA\y}_G$ for all $\bA\y\in\mathcal{N}$, with probability at least $1-\frac{1}{\poly(n)}$.

We claim accuracy on this $\eps$-net suffices to prove accuracy everywhere. 
Let $\z\in\mathbb{R}^d$ be a nonzero normalized vector so that $\norm{\bA\z}_G=1$. 
We inductively define a sequence $\bA\y_1,\bA\y_2,\ldots$ with $\norm{\bA\z-\sum_{j=1}^i\bA\y_j}_G\le\eps^i$ and there exists some constant $\gamma_i\le\eps^{i-1}$ with $\frac{1}{\gamma_i}\bA\y_i\in\mathcal{N}$ for all $i$. 
For the base case, we define $\bA\y_1$ to be the closest point to $\bA\z$ in the $\eps$-net $\mathcal{N}$ so that $\norm{\bA\z-\bA\y_1}_G\le\eps$. 
For the inductive step, given a sequence $\bA\y_1,\ldots,\bA\y_{i-1}$ such that $\gamma_i:=\norm{\bA\z-\sum_{j=1}^{i-1}\bA\y_j}_G\le\eps^{i-1}$, note that $\frac{1}{\gamma_i}\norm{\bA\z-\sum_{j=1}^{i-1}\bA\y_j}_G=1$, the next point in the sequence is defined as $\bA\y_i\in\mathcal{N}$ so that $\bA\y_i$ is within distance $\eps$ of $\bA\z-\sum_{j=1}^{i-1}\bA\y_j$. 
Hence,
\[|\norm{\bM\z}_G-\norm{\bA\z}_G|\le\sum_{i=1}^\infty|\norm{\bM\y_i}_G-\norm{\bA\y_i}_G|\le\sum_{i=1}^\infty\eps^i\norm{\bA\y_i}_G=\O{\eps}\norm{\bA\z}_G,\]
which finally completes the inductive step for time $n$. 
\end{proof}

\begin{lemma}
\lemlab{lem:sum:sensitivities}
\cite{BravermanDMMUWZ20}
For a matrix $\bA=\a_1\circ\ldots\circ\a_n\in\mathbb{R}^{n\times d}$ that arrives as a stream with condition number $\kappa$, let $\ell_i$ be the online $L_1$ sensitivity of $\a_i$. 
Then $\sum_{i=1}^n\ell_i=\O{d\log n\log\kappa}$. 
\end{lemma}

\begin{theorem}[Subspace Embedding for Orlicz Norms in the Streaming Model]
\thmlab{thm:orlicz:se}
Given $\eps>0$ and a matrix $\bA\in\mathbb{R}^{n\times d}$ whose rows $\a_1,\ldots,\a_n$ arrive sequentially in a stream with condition number at most $\kappa$, there exists a streaming algorithm that outputs a $(1+\eps)$ subspace embedding for an Orlicz norm with high probability. 
The algorithm samples $\O{\frac{d^2\Delta}{\eps^2}\log^2 n\log\kappa}$ rows, with high probability, where $\Delta$ is defined as in \lemref{lem:subspace}. 
\end{theorem}
\begin{proof}
\algref{alg:orlicz} is correct with high probability, by \lemref{lem:subspace}. 
It remains to analyze the space complexity of \algref{alg:orlicz}. 
By \lemref{lem:subspace} and a union bound over the $n$ rows in the stream, each row $\a_i$ is sampled with probability at most $4\alpha\Delta\ell_i$, where $\ell_i$ is the online leverage score of row $\a_i$.   
By \lemref{lem:sum:sensitivities}, $\sum_{i=1}^n\ell_i=\O{d\log n\log\kappa}$. 
Since $\alpha=\O{\frac{d}{\eps^2}\log n}$, then by a standard coupling and concentration argument, the total number of sampled rows is $\O{\frac{d^2\Delta}{\eps^2}\log^2 n\log\kappa}$. 
\end{proof}

\paragraph{Applications to the Sliding Window Model.}
Coresets are dimensionality reduction tools with extensive applications~\cite{FeldmanL11,BachemLK17,SohlerW18,MunteanuSSW18,AssadiBBMS19,BaykalLGFR19,BravermanLUZ19,HuangV20,Feldman20,MussayOBZF20,MahabadiRWZ20,BravermanHMSSZ21}. 
An online coreset for a matrix $\bA$ is a weighted subset of rows of $\bA$ that also provides a good approximation to a certain desired function (such as Orlicz norm) to all prefixes of $\bA$. 
\begin{definition}[Online Coreset]
An \emph{online coreset} for a function $f$, an approximation parameter $\eps>0$, and a matrix $\bA\in\mathbb{R}^{n\times d}=\a_1\circ\ldots\circ\a_n$ is a subset of weighted rows of $\bA$ such that for any $\bA_i=\a_1\circ\ldots\circ\a_i$ with $i\in[n]$, $f(\bM_i)$ is a $(1+\eps)$-approximation of $f(\bA_i)$, where $\bM_i$ is the matrix that consists of the weighted rows of $\bA$ in the coreset that appear at time $i$ or before. 
\end{definition}
Observe that the proof of \lemref{lem:subspace} is by induction and thus shows that \algref{alg:orlicz} admits an online coreset. 
\cite{BravermanDMMUWZ20} showed that a streaming algorithm that admits an online coreset with probability $1-\frac{1}{\poly(n)}$ can be adapted to a sliding window algorithm for $W=\poly(n)$.  
\begin{theorem}
\thmlab{thm:coreset:framework}
\cite{BravermanDMMUWZ20}
Let $\r_1,\ldots,\r_n\in\mathbb{R}$ be a stream of rows, $\eps>0$, and $\bA=\r_{n-W+1}\circ\ldots\circ\r_n$ be the matrix consisting of the $W$ most recent rows. 
If there exists a online coreset algorithm for a matrix function $f$ that stores $S(n,d,\eps)$ rows, then there exists a sliding window algorithm that stores $\O{S\left(n,d,\frac{\eps}{\log n}\right)\log n}$ rows and outputs a matrix $\bM$ such that $f(\bM)$ is a $(1+\eps)$-approximation of $f(\bA)$. 
\end{theorem}
Since \algref{alg:orlicz} admits an online coreset with probability $1-\frac{1}{\poly(n)}$, then \thmref{thm:coreset:framework} implies a sliding window algorithm for Orlicz norms:
\begin{theorem}[Subspace Embedding for Orlicz Norms in the Sliding Window Model]
\thmlab{thm:orlicz:se:sw}
Given $\eps>0$ and a matrix $\bA\in\mathbb{R}^{W\times d}$ whose rows $\a_1,\ldots,\a_W$ arrive sequentially in a stream $\r_1,\ldots,\r_n$ with condition number at most $\kappa$, there exists a sliding window algorithm that outputs a $(1+\eps)$ subspace embedding for an Orlicz norm with high probability. 
The algorithm samples $\O{\frac{d^2\Delta}{\eps^2}\log^5 n\log\kappa}$ rows, with high probability, where $\Delta$ is defined as in \lemref{lem:subspace}. 
\end{theorem}

\section{Approximate Heavy-Hitters in the Sliding Window Model}
In this section, we describe our $\nu$-approximate $\eta$-heavy hitters algorithm that appears in \algref{alg:hh}, slightly perturbing constants for the ease of discussion. 
%Our heavy-hitters algorithm will be a crucial subroutine for the approximation of general symmetric norms on sliding windows in \secref{sec:sym:norm}. 
Our starting point is the $L_2$ norm estimation algorithm $\ams$ in~\cite{BravermanO07}. 
$\ams$ maintains a number of timestamps $\{t_i\}$ throughout the data stream, along with a separate streaming algorithm for each $t_i$ that stores a sketch of the $L_2$ norm of the elements in the stream after $t_i$. 
\cite{BravermanO07} observes that it suffices for $\{t_i\}$ to maintain the invariant that the sketches of at most two timestamps produce values that are within $2$ of each other, since by \lemref{lem:count} (and with precise constants) they would always output values that are within $2$ afterwards. 
Hence, if the length of the stream $m$ is polynomially bounded in $n$, then the number of total timestamps is $\O{\log n}$. 
Moreover, two of these timestamps will sandwich the starting point of the sliding window and provide a $2$-approximation to the $L_2$ norm of the active elements and more generally, there exists an algorithm $\ams$ that outputs a $2$-approximation to \emph{any} suffix of the stream. 
See \figref{fig:histogram} for an illustration about the intuition of $\ams$. 
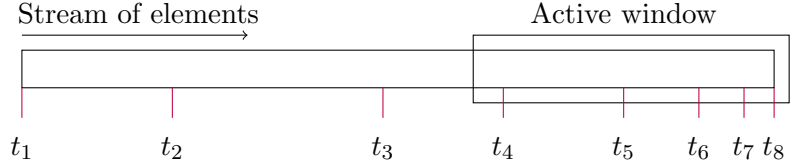
\begin{figure*}[!htb]
\centering
\begin{tikzpicture}
\draw (0,0) rectangle+(10,0.5);
\draw (6,-0.2) rectangle+(4.2,0.9);
\draw[purple] (10,0)--(10,-0.4);
\node at (10,-0.8){$t_8$};
\draw[purple] (9.6,0)--(9.6,-0.4);
\node at (9.6,-0.8){$t_7$};
\draw[purple] (9,0)--(9,-0.4);
\node at (9,-0.8){$t_6$};
\draw[purple] (8,0)--(8,-0.4);
\node at (8,-0.8){$t_5$};
\draw[purple] (6.4,0)--(6.4,-0.4);
\node at (6.4,-0.8){$t_4$};
\draw[purple] (4.8,0)--(4.8,-0.4);
\node at (4.8,-0.8){$t_3$};
\draw[purple] (2,0)--(2,-0.4);
\node at (2,-0.8){$t_2$};
\draw[purple] (0,0)--(0,-0.4);
\node at (0,-0.8){$t_1$};

\draw[->](0,0.7) -- (3,0.7);
\node at (1.55,1){Stream of elements};
\node at (8,1){Active window};
\end{tikzpicture}
\caption{$\ams$ maintains a series of timestamps $t_i$ along with a sketch of the $L_2$ norm of the updates from time $t_i$ to the end of the stream. 
The timestamps have the invariant that at most two sketches output values that are within $2$ of each other. 
In particular, $t_3$ and $t_4$ sandwich the $L_2$ norm of the active window within a factor of $2$.}
\figlab{fig:histogram}
\end{figure*}

To transition from $L_2$ norm estimation to $\eta$-heavy hitters, \cite{BravermanGO14,BravermanGLWZ18} simultaneously run instances of the $\countsketch$ heavy-hitter algorithm starting at each of the timestamps $t_i$. 
Any $\eta$-heavy hitter of the active elements must be a $\frac{\eta}{2}$-heavy hitter of the stream starting at some timestamp, since one of these timestamps $t_i$ contains the active elements but has $L_2$ norm at most $2$ times the $L_2$ norm of the active elements. 
Hence, all $\eta$-heavy hitters will be reported by the corresponding $\countsketch$ starting at $t_i$. 
However, it can also report elements that do not appear in the window at all, e.g., the elements after $t_i$ but before $m-W+1$. 
Thus, \cite{BravermanGO14,BravermanGLWZ18} also maintains a constant factor approximation to the frequency of each item reported by $\countsketch$ as a final check, through comparison with the estimated $L_2$ norm from $\ams$. 
These parameters are insufficient to obtain $\nu$-approximate $\eta$-heavy hitters, since 1) a constant factor approximation to each frequency cannot give a $(1+\nu)$-approximation and 2) if $\countsketch$ only reports elements once they are $\eta$-heavy, then it is possible that a constant fraction of the frequency is missed, e.g., if the frequency is $2\eta\cdot\norm{f}_2$. 
To address these issues, we apply two simple fixes in \figref{fig:crude}. 

\begin{figure}[!htb]
\begin{mdframed}
\begin{enumerate}
\item
Find a superset of the possible heavy-hitters of the active window by taking heavy-hitters of a superset of the active window, but with a lower threshold, i.e. $\O{\nu\eta}$ rather than $\eta$. 
\item
For each possible heavy-hitter, maintain a $(1+\O{\nu})$-approximation to its frequency. 
\item
Report the items with sufficiently high estimated frequency.
\end{enumerate}
\end{mdframed}
\vspace{-0.1in}
\caption{Crude outline of $\nu$-approximate $\eta$-heavy hitter sliding window algorithm.}
\figlab{fig:crude}
\end{figure}

First, we maintain a $(1+\O{\nu})$-approximation to the frequency of each item reported by $\countsketch$. 
However, we note that we only track the frequency of an item once it is reported by $\countsketch$ and thus the second issue still prevents our algorithm from reporting a $(1+\nu)$-approximation for sufficiently small $\nu$ because a constant fraction of the frequency can still be missed before being reported by $\countsketch$. 
Thus the second idea is to report items once they are $\frac{\nu\eta}{32}$-heavy hitters, so that only a $\O{\nu}$ fraction of the frequency can be missed before each heavy-hitter is tracked. 
We give a crude outline of our approach in \figref{fig:crude} and the algorithm in full in \algref{alg:hh}. 

\begin{algorithm}[!htb]
\caption{Algorithm for $\eta$-heavy hitters in sliding window model, with $(1+\nu)$-approximation to frequency of reported items.}
\alglab{alg:hh}
\begin{algorithmic}[1]
\Require{A stream of elements $u_1,\ldots,u_m\in[n]$, a window parameter $W>0$, threshold $\eta\in(0,1)$, accuracy parameter $\nu\in\left(0,\frac{1}{4}\right)$}
\Ensure{A list that contains all $\eta$-heavy hitters and no element that is not a $\frac{\eta}{2}$-heavy hitters, along with a $(1+\nu)$ to the frequency of all items.}
\State{Run an instance of $\ams$ on the stream.}
\State{$\T\gets\emptyset$}
\For{each update $u_t\in[n]$ with $t\in[m]$}
\State{$\T\gets\T\cup\{t\}$}
\State{Initialize $\countsketch_t$ with threshold $\frac{\nu\eta}{32}$.}
\Comment{Identify a superset of the heavy-hitters}
\State{$X_a\gets$ estimated $L_2$ norm of the frequency vector from time $a\in\T$ to $t$ by $\ams$.}
\While{exist $b<c\in\T$ with $c<t-W+1$ or $a<b<c\in\T$ with $X_a\le\frac{17}{16}X_c$}
\State{Delete $b$ from $\T$ and $\countsketch_b$.}
\EndWhile
\State{$H_a\gets$ heavy-hitters reported by $\countsketch_a$ from time $a\in\T$ to $t$.}
\State{$F\gets$ estimated $L_2$ norm of the frequency vector from time $\min(1,t-W+1)$ to $t$ by $\ams$.}
\For{all $a\in\T$ and $i\in H_a$}
\State{Use $\counter$ for $i$, starting at time $a$.}
\Comment{$\left(1+\frac{\nu}{4}\right)$-accuracy}
\State{$\widehat{f_i}\gets$ any underestimate to the frequency of $i$ in the last $W$ updates by $\counter$.} 
\If{$\widehat{f_i}\ge\frac{\eta}{2}\cdot F$}
%\Comment{Check if frequency might exceed threshold}
\label{line:freq:check}
\State{Report $i$, with estimated frequency $\widehat{f_i}$}
\EndIf
\EndFor
\EndFor
\end{algorithmic}
\end{algorithm}

We first show that \algref{alg:hh} does not output any items with sufficiently low frequency. 
\begin{lemma}[Low frequency items are not reported]
\lemlab{lem:hh:low}
Let $f$ be the frequency vector induced by the active window. 
For each $i\in[n]$, if $f_i\le\frac{\eta}{8}\norm{f}_2$, then \algref{alg:hh} does not report $i$. 
\end{lemma}
\begin{proof}
Observe that either (1) $i\in\cup_{a\in\T}H_a$, so that $i$ is a $\frac{\nu\eta}{32}$-heavy hitter of some suffix of the stream, or (2) $i\in\cup_{a\in\T}H_a$. 
In the latter case, $i\notin H_a$ for any $a\in\T$, so then $i$ will not be reported. 
In the former case, an instance of $\counter$ is maintained for $i$, so that $\widehat{f_i}$ is an underestimate of $f_i$. 
But $F$ is a $C$-approximation to $\norm{f}_2$ with $C=\frac{17}{16}$, so then $f_i\le\frac{\eta}{8}\norm{f}_2$ implies $\widehat{f_i}<\frac{\eta}{2}\cdot F$ for $\nu\in\left(0,\frac{1}{4}\right)$. 
Thus, $i$ will not be reported due to Line~\ref{line:freq:check} of \algref{alg:hh}. 
\end{proof}

Next we show that not only are the heavy-hitters reported, but the estimated frequency for each reported item is also a $(1+\nu)$ approximation to the true frequency. 
\begin{lemma}[Heavy-hitters are reported accurately]
\lemlab{lem:hh:high}
Let $f$ be the frequency vector induced by the active window. 
For each $i\in[n]$, if $f_i\ge\eta\cdot\norm{f}_2$, then \algref{alg:hh} reports $i$. 
Moreover, $\widehat{f_i}\le f_i\le(1+\nu)\widehat{f_i}$ for any item $i$ reported by \algref{alg:hh}. 
\end{lemma}
\begin{proof}
For $\nu\in\left(0,\frac{1}{4}\right)$, the condition $f_i\ge\eta\cdot\norm{f}_2$ implies that $i$ is a $\nu\eta$-heavy hitter of some suffix of the stream. 
Namely for a stream of length $m$, let $a\in\T$ with $a\le m-W+1$, so that the $L_2$ norm of the underlying frequency induced by the updates starting from time $a$ is a $C$-approximation of $\norm{f}_2$, with $C=\frac{17}{16}$. 
Then once $\frac{\nu\eta}{16}\cdot\norm{f}_2$ instances of $i$ are inserted at some time $t$ after $m-W+1$, $i$ will \emph{always} be reported as a $\frac{\nu\eta}{32}$-heavy hitter by $\countsketch_a$. 
Hence, $i\in H_a$ and an instance of $\counter$ is maintained for $i$, starting at time $t$. 
Since at most $\frac{\nu\eta}{16}\cdot\norm{f}_2$ instances of $i$ arrive before $t$, certainly at least $\left(1-\frac{\nu}{2}\right)\eta\cdot\norm{f}_2$ instances of $i$ remain after $t$. 
Hence for $\nu\in\left(0,\frac{1}{4}\right)$, $\counter$ reports at least 
\[\left(\frac{1}{1+\frac{\nu}{4}}\right)\left(1-\frac{\nu}{2}\right)\eta\cdot\norm{f}_2\ge\left(\frac{1}{1+\nu}\right)\eta\cdot\norm{f}_2\ge\frac{\eta}{2}\cdot F\]
instances of $i$ in the active window, since $F\le\norm{f}_2$. 
Thus, $i$ passes the check of Line~\ref{line:freq:check} and is reported by \algref{alg:hh}, which completes the first part of the claim.

Similarly, at most $\frac{\nu\eta}{16}\cdot\norm{f}_2$ instances of \emph{any} item $i$ reported by \algref{alg:hh} can be missed before $i$ is reported as a $\frac{\nu\eta}{32}$-heavy hitter. 
If $i$ passes the check of Line~\ref{line:freq:check} and is reported by \algref{alg:hh}, then 
\[\widehat{f_i}\ge\frac{\eta}{2}\cdot F\ge\frac{\eta}{4}\cdot\norm{f}_2,\]
since $F\ge2\norm{f}_2$. 
Thus, $f_i\ge\frac{\eta}{4\left(1+\frac{\nu}{4}\right)}\cdot\norm{f}_2$. 
Hence the additive error in the estimation of $f_i$ due to the missing $\frac{\nu\eta}{16}\cdot\norm{f}_2$ instances of $i$ is only a relative $\frac{\nu\left(1+\frac{\nu}{4}\right)}{4}$ error. 
Combined with the relative $\left(1+\frac{\nu}{4}\right)$ error of $\counter$, the total relative error is at most $(1+\nu)$, for $\nu\in\left(0,\frac{1}{4}\right)$. 
\end{proof}
Finally, we justify the correctness of our $\nu$-approximate $\eta$-heavy hitters sliding window algorithm and analyze the space complexity. 
\begin{proofof}{\thmref{thm:hh}}
Consider \algref{alg:hh}. 
Observe that the correctness guarantees of the algorithm follow immediately from \lemref{lem:hh:low} and \lemref{lem:hh:high}. 
The space complexity follows from noting that for $\log m=\O{\log n}$, the $L_2$ norm of the underlying vector of the entire stream is polynomially bounded in $n$. 
Thus, there are at most $\O{\log n}$ times in $\T$ by \lemref{lem:count}. 
For each time $a\in\T$, \algref{alg:hh} uses an instance of $\countsketch$ with threshold $\nu\eta$, an instance of $\ams$, and an instance of $\counter$ for each heavy-hitter reported by $\countsketch_a$. 
By \thmref{thm:countsketch}, each instance of $\countsketch$ uses $\O{\frac{1}{\nu^2\eta^2}\log^2 n}$ bits of space. 
By \lemref{lem:ams}, each instance of $\ams$ uses $\O{\log^2 n}$ bits of space. 
Each $\countsketch_a$ with threshold $\nu\eta$ can report up to $\O{\frac{1}{\nu^2\eta^2}}$ items, and each instance of $\counter$ uses $\O{\frac{1}{\nu}\log^2 n}$ bits of space by \lemref{lem:counter}. 
Thus, the total space used by \algref{alg:hh} is $\O{\frac{1}{\nu^3\eta^2}\log^3 n}$ bits. 
\end{proofof}

\section{Symmetric Norms}
\seclab{sec:sym:norm}
In this section, we formalize our symmetric norm sliding window algorithm and give a number of applications.   
We first require the following preliminary definitions that quantify specific properties of symmetric norms. 
\begin{definition}[Modulus of concentration]
Let $X\in\mathbb{R}^n$ be a random variable uniformly distributed on the $L_2$-unit sphere $S^{n-1}$. 
The median of a symmetric norm $\ell$ is the unique value $\M_{\ell}$ such that $\PPr{\ell(X)\ge\M_{\ell}}\ge\frac{1}{2}$ and $\PPr{\ell(X)\le\M_{\ell}}\ge\frac{1}{2}$. 
Then if $\b_{\ell}$ denotes the maximum value of $\ell(x)$ over $x\in S^{n-1}$, then the ratio $\mc(\ell):=\frac{\b_{\ell}}{\M_{\ell}}$ is called the \emph{modulus of concentration} of the norm $\ell$. 
\end{definition}
The modulus of concentration characterizes the average behavior of the norm $\ell$ on $\mathbb{R}^n$. 
However, even if $\ell$ is well-behaved on average, more difficult norms can be embedded and hidden in a lower-dimensional subspace. 
For example, \cite{BlasiokBCKY17} observes that $\mc(\ell)=\O{1}$ for the $L_1$ norm $\ell$, but when $x$ has fewer than $\sqrt{n}$ nonzero coordinates, the norm $\ell(x)=\max(L_{\infty}(x),L_1(x)/\sqrt{n})$ on the unit ball becomes identically $L_{\infty}(x)$, which requires $\Omega(\sqrt{n})$ space~\cite{AlonMS99}. 
Thus, we instead consider the modulus of concentration over all lower dimensions.
\begin{definition}[Maximum modulus of concentration]
For every $k\le n$, the norm $\ell:\mathbb{R}^n\to\mathbb{R}$ induces a norm on $\mathbb{R}^k$ by setting $\ell^{(k)}((x_1,\ldots,x_k))=\ell((x_1,\ldots,x_k,0,\ldots,0))$.
The \emph{maximum modulus of concentration} of the norm $\ell$ is defined as $\mmc(\ell):=\underset{k\le n}{\max}\mc(\ell^{(k)})=\underset{k\le n}{\max}\frac{\b_{\ell^{(k)}}}{\M_{\ell^{(k)}}}$.
\end{definition}
We now reduce the problem of approximating a symmetric norm $\ell$ to the $\nu$-approximate $\eta$-heavy hitters problem.
\begin{lemma}[Symmetric norm approximation through heavy-hitters]
\lemlab{lem:hh:to:out}
\cite{BlasiokBCKY17}
%Let $\ell$ be any symmetric norm and $\eps>0$ be a fixed accuracy parameter. 
%Let $\nu:=\O{\frac{\eps^2}{\log n}}$ be an additional fixed accuracy parameter and $\eta:=\O{\frac{\eps^{5/2}}{\mmc(\ell)\log^{5/2} n}}$ be a fixed threshold, both of which are dependent on $\ell$. 
%For each $i\in[\log n]$, let $j\in[n]$ be sampled into $S_i$ with probability $\frac{1}{2^i}$. 
%Let $f$ be a frequency vector on $[n]$ and for each $i\in[\log n]$, let $g_i$ be the frequency vector induced by setting all coordinates $j\in[n]$ of $f$ with $j\notin S_i$. 
%Suppose $H_i$ is a set of lists, such that for each $i\in[n]$:
%\begin{enumerate}
%\item
%If $f_j\ge\eta\norm{g_i}_2$ for any $j\in[n]$, then $j\in H_i$. 
%That is, $H_i$ contains all $\eta$-heavy hitters of $g_i$. 
%\item
%If $f_j\le\frac{\eta}{2}\norm{g_i}_2$ for any $j\in[n]$, then $j\notin H_i$. 
%That is, $H_i$ does not contain any item that is not an $\frac{\eta}{2}$-heavy hitter of $g_i$. 
%\item
%If $j\in H_i$, then $H_i$ reports a value $\widehat{f_j}$ such that $(1-\nu)f_j\le\widehat{f_j}\le f_j$. 
%That is, $H_i$ outputs a $(1+\nu)$-approximation to the frequency $f_j$ of all $j\in H_i$. 
%\end{enumerate}
%Then there exists an algorithm $\estimate$ that takes inputs $\nu$, $\eta$, and $\{H_i\}$ and outputs a $(1+\eps)$-approximation to $\ell(f)$. 
%Moreover, the running time of $\estimate$ is polynomial in $\frac{1}{\eps}$ and $n$ and the working space of $\estimate$ is the space used to store $\{H_i\}$. 
Let $\ell$ be any symmetric norm, $\eps>0$ and $\nu:=\O{\frac{\eps^2}{\log n}}$ be fixed accuracy parameters, and $\eta:=\O{\frac{\eps^{5/2}}{\mmc(\ell)\log^{5/2} n}}$ be a fixed threshold. 
Let $R=\Theta\left(\frac{\log^{10}n}{\eps^5}\right)$ and for each $i\in[\log n]$ and $r\in[R]$, let $j\in[n]$ be sampled into $S_{i,r}$ with probability $\frac{1}{2^i}$. 
Let $f$ be a frequency vector (possibly implicitly) defined on $[n]$ and for each $i\in[\log n]$, let $g_{i,r}$ be the frequency vector induced by setting all coordinates $j\in[n]$ of $f$ with $j\notin S_i$. 

Suppose there exists an algorithm that outputs $\nu$-approximate $\eta$-heavy hitters $H_{i,r}$ for each $g_{i,r}$. 
There exists a recovery function $\estimate$ that recovers a $(1+\eps)$-approximation to $\ell(f)$ using $\{H_{i,r}\}$. 
The running time of $\estimate$ is polynomial in $\frac{1}{\eps}$ and $n$ and the working space of $\estimate$ is the space used to store $\{H_i\}$. 
\end{lemma}
Informally, \lemref{lem:hh:to:out} states that to obtain a $(1+\eps)$-approximation to any symmetric norm $\ell$ of an underlying frequency vector, it suffices to use a $\nu$-approximate $\eta$ heavy-hitter algorithm. 
Here, $\eta$ and $\nu$ are parameters dependent on the norm $\ell$. 
We give additional intuition into \lemref{lem:hh:to:out} and its proof by \cite{BlasiokBCKY17} in \appref{app:estimate}. 

We now formalize the argument of \thmref{thm:sym:norm} by derandomizing the space complexity and analyzing the space complexity. 
The proof of \thmref{thm:sym:universal} is identical. 
\begin{proofof}{\thmref{thm:sym:norm}}
Let $\ell$ be any symmetric norm, $\eps>0$ and $\nu:=\O{\frac{\eps^2}{\log n}}$ be fixed accuracy parameters, and $\eta:=\O{\frac{\eps^{5/2}}{\mmc(\ell)\log^{5/2} n}}$ be a fixed threshold. 
Let $R=\Theta\left(\frac{\log^{10}n}{\eps^5}\right)$ and for each $i\in[\log n]$ and $r\in[R]$, let $j\in[n]$ be sampled into $S_{i,r}$ with probability $\frac{1}{2^i}$. 
Let $f$ be a frequency vector (possibly implicitly) defined on $[n]$ and for each $i\in[\log n]$, let $g_{i,r}$ be the frequency vector induced by setting all coordinates $j\in[n]$ of $f$ with $j\notin S_i$. 
Suppose we run an instance $A_{i,r}$ of our $\nu$-approximate $\eta$-heavy hitters algorithm \algref{alg:hh} for each $g_{i,r}$. 
By \thmref{thm:hh}, we obtain the $\nu$-approximate $\eta$-heavy hitters $H_{i,r}$ for each $g_{i,r}$. 
Thus by \lemref{lem:hh:to:out}, we can use $H_{i,r}$ and the procedure $\estimate$ to recover a $(1+\eps)$-approximation to $\ell(f)$. 

%\paragraph{Space complexity.} 
To analyze the space complexity, observe that we have a single instance of the $\nu$-approximate $\eta$-heavy hitters algorithm \algref{alg:hh} for each $g_{i,r}$, with $\nu:=\O{\frac{\eps^2}{\log n}}$, $\eta:=\O{\frac{\eps^{5/2}}{\mmc(\ell)\log^{5/2} n}}$. 
By \thmref{thm:hh}, the total space required for each instance of \algref{alg:hh} is $\O{\frac{\log^8 n}{\eps^{11}}\mmc(\ell)^2}$. 
Since $i\in[\log n]$, $r\in[R]$, and $R=\Theta\left(\frac{\log^{10}n}{\eps^5}\right)$, then the total space for the algorithm is $\O{\frac{\log^{19} n}{\eps^{16}}\mmc(\ell)^2}$, given unlimited access to random bits. 
Finally, if we use Nisan's PRG to derandomize our algorithm, then the total space for the algorithm is $\O{\frac{\log^{20} n}{\eps^{16}}\mmc(\ell)^2}$. 
We give full details on the derandomization in \appref{app:rng}.  
\end{proofof}

\subsection{Applications}
\seclab{sec:apps}
In this section, we demonstrate the application of \thmref{thm:sym:norm} and \thmref{thm:sym:universal} to a number of symmetric norms. 
We summarize our results in \figref{fig:results}.

\paragraph{$Q'$-norms.}
We first that a $(1+\eps)$-approximation of any $Q'$-norm, i.e., quadratic norm, in the sliding window model only requires polylogarithmic space, using the maximum modulus of concentration characterization of $Q$-norms by \cite{BlasiokBCKY17}. 
\begin{definition}[$Q$-norm and $Q'$-norm]
A norm $\ell:\mathbb{R}^n\to\mathbb{R}$ is a $Q$-norm if there exists a symmetric norm $L:\mathbb{R}^n\to\mathbb{R}$ such that for all $x\in\mathbb{R}^n$, we have $\ell(x)=L(x^2)^{1/2}$, where $x^2$ denotes the coordinate-wise square power of $x$. 
Then a norm $\ell':\mathbb{R}^n\to\mathbb{R}$ is a $Q'$-norm if its dual norm is a $Q$-norm. 
\end{definition}
$Q'$-norms includes the $L_p$ norms for $1\le p\le 2$. 
\cite{BlasiokBCKY17} also notes that multiple $Q'$-norms have been proposed to regularize sparse recovery problems in machine learning. 
For example, \cite{ArgyriouFS12} shows that the $k$-support norm, whose unit ball is the convex hull of the set $\{x\in\mathbb{R}^n:\norm{x}_0\le k\text{ and }\ell_2(x)\le 1\}$, is a $Q'$-norm that has a tighter relaxation than elastic nets and can thus be more effective for sparse prediction. 
The box norm~\cite{McDonaldPS14}, defined for $\Theta=\{\theta\in[a,b]^n\,:\,\ell_1(x)\le c\}$, given parameters $0<a<b\le c$, as $\ell_{\Theta}(x)=\min_{\theta\in\Theta}\left(\sum_{i=1}^n x_i^2/\theta_i\right)^{1/2}$, is a $Q'$-norm that is also a generalization of the $k$-support norm. 
The box norm has been used to further optimize algorithms for the sparse prediction problem specifically in the context of multitask clustering~\cite{McDonaldPS14}. 
\begin{lemma}
\lemlab{lem:mmc:qnorm}
\cite{BlasiokBCKY17}
$\mmc(\ell)=\O{\log n}$ for every $Q'$-norm $\ell$. 
\end{lemma}
\noindent
From \thmref{thm:sym:norm} and \lemref{lem:mmc:qnorm}, we obtain a sliding window algorithm for $Q'$-norm estimation. 
\begin{corollary}
\corlab{cor:qnorm}
Given $\eps>0$, there exists a sliding window algorithm that uses $\poly\left(\frac{1}{\eps},\log n\right)$ bits of space and outputs a $(1+\eps)$-approximation to the $Q'$-norm. 
\end{corollary}

\paragraph{$L_p$ norms.}
Since $Q'$-norms include $L_p$ norms for $p\in[1,2]$, we now consider the approximation of $L_p$ norms for $p>2$. 
\begin{lemma}
\lemlab{lem:mmc:lp}
\cite{BlasiokBCKY17}
$\mmc(\ell)=\O{n^{1/2-1/p}}$ for every $L_p$ norm with $p>2$. 
\end{lemma}
Thus \thmref{thm:sym:norm} and \lemref{lem:mmc:lp} implies the following sliding window algorithm for $L_p$-norm estimation. 
\begin{corollary}
\corlab{cor:lp}
Given $\eps>0$ and $p>2$, there exists a sliding window algorithm that uses $\poly\left(\frac{1}{\eps},\log n\right)\cdot n^{1-2/p}$ bits of space and outputs a $(1+\eps)$-approximation to the $L_p$-norm. 
\end{corollary}
In particular, since the exponents of $\eps$ and $\log n$ are fixed, then for sufficiently large $p$, \corref{cor:lp} improves on the results of \cite{BravermanO07}, who give an algorithm using space $\frac{1}{\eps^{p+2}}\,\polylog n\cdot n^{1-2/p}$.

\paragraph{Top-$k$ norms.}
We now show that a $(1+\eps)$-approximation of any top-$k$ norm in the sliding window model only requires sublinear space, for sufficiently large $k$. 
\begin{definition}[Top-$k$ norm]
The top-$k$ norm for a vector $x\in\mathbb{R}^n$ is the sum of the largest $k$ coordinates of $|x|$. 
\end{definition}
The top-$k$ norm is a special case of the Ky Fan $k$-norm~\cite{WuDST14} when the vector $x$ represents the entries in a diagonal matrix. 
Thus the top-$k$ norm is often used to understand the Ky Fan $k$-norm, which is used to regularize optimization problems in numerical linear algebra. 
\begin{lemma}
\lemlab{lem:mmc:topk}
\cite{BlasiokBCKY17}
$\mmc(\ell)=\tO{\sqrt{\frac{n}{k}}}$ for the top-$k$ norm $\ell$. 
\end{lemma}
\noindent
From \thmref{thm:sym:norm} and \lemref{lem:mmc:topk}, we obtain a sliding window algorithm for top-$k$ norm estimation. 
\begin{corollary}
\corlab{cor:topk}
Given $\eps>0$, there exists a sliding window algorithm that uses $\frac{n}{k}\cdot\poly\left(\frac{1}{\eps},\log n\right)$ bits of space and outputs a $(1+\eps)$-approximation to the top-$k$ norm.  
\end{corollary}

\section*{Acknowledgments}
Samson Zhou was supported by a Simons Investigator Award of David P. Woodruff.

\def\shortbib{0}
\bibliographystyle{alpha}
\bibliography{references}

\appendix
\section{Derandomization of \algref{alg:hh}}
\applab{app:rng}
We first require the following pseudorandom generator to derandomize our algorithms.
\begin{theorem}[Nisan's PRG]
\cite{Nisan92}
\thmlab{thm:nisan}
Let $\mathcal{A}$ be an algorithm that uses $S=\Omega(\log n)$ space and $R$ random bits. 
Then there exists a pseudorandom generator for $\mathcal{A}$ that succeeds with high probability and runs in $\O{S\log R}$ bits. 
\end{theorem}
We now claim the correctness of the derandomization of our algorithm using Nisan's PRG. 
Recall that Nisan's PRG can be viewed as generating a stream of pseudorandom bits in a read-once tape that can be used to generate random variables to fool a small space tester. 
However, an input tape that can only be read once cannot be immediately given to algorithm to generate the randomness required for the hash functions that govern whether an index $j\in[n]$ is sampled into the sets $S_{i,r}$ in \algref{alg:hh} because the indices sampled by each $S_{i,r}$ must be consistent whenever each coordinate of the frequency vector $i$ is updated. 
Instead, we use the standard reordering trick to derandomize using Nisan's PRG and argue indistinguishability. 

For any fixed randomness $\R$ for the sampling of the set $S_{i,r}$, let $\T_{\R}$ be the tester that tests whether our heavy-hitter algorithm would output an index $j\in[n]$ if $\R$ is hard-coded into the tester and the random bits for the sampling procedures arrive in the stream. 
Formally, we define $\T_{\R}(j,\S,\A_1)=1$ if the algorithm with access to independent random bits outputs $i$ on stream $\S$ and $\T_{\R}(j,\S,\A_1)=0$ otherwise. 
Similarly, we define $\T_{\R}(j,\S,\A_2)=1$ if using Nisan's PRG on our algorithm outputs $i$ on stream $\S$ and $\T_{\R}(j,\S,\A_2)=0$ otherwise. 

For any fixed input stream $\S_1$, let $\S_2$ be an input stream in which all updates to a single coordinate of the underlying frequency vector arrive consecutively in the active window. 
Nisan's PRG on the algorithm suffices to fool the tester $\T_{\R}$ on $\S_2$ from an algorithm with unlimited access to random bits, i.e., $
\Bigg|\PPr{\T_{\R}(j,\S_2,\A_1)=1}-\PPr{\T_{\R}(j,\S_2,\A_2)=1}\Bigg|=\frac{1}{\poly(n)}$,
%\begin{align*}
%\Bigg|\PPr{\T_{\R}(j,\S_2,\A_1)=1}-&\PPr{\T_{\R}(j,\S_2,\A_2)=1}\Bigg|\\
%&=\frac{1}{\poly(n)}.
%\end{align*}
for all $j\in[n]$. 
On the other hand, the order of the inputs does not change the identity of the heavy-hitters within the active window, so that $\PPr{\T_{\R}(j,\S_1,\A_1)=1}=\PPr{\T_{\R}(j,\S_2,\A_1)=1}$. 
Similarly, the order of the inputs does not change the identity of the heavy-hitters within the active window following Nisan's PRG, so that $\PPr{\T_{\R}(j,\S_1,\A_2)=1}=\PPr{\T_{\R}(j,\S_2,\A_2)=1}$. 
Thus, $\Bigg|\PPr{\T_{\R}(j,\S_1,\A_1)=1}-\PPr{\T_{\R}(j,\S_1,\A_2)=1}\Bigg|=\frac{1}{\poly(n)}$,
%\begin{align*}
%\Bigg|\PPr{\T_{\R}(j,\S_1,\A_1)=1}-&\PPr{\T_{\R}(j,\S_1,\A_2)=1}\Bigg|\\
%&=\frac{1}{\poly(n)},
%\end{align*}
so that with high probability, a tester cannot distinguish between an algorithm with derandomization using Nisan's PRG and unlimited access to random bits. 
The argument is completed by union bounding over all indices $j\in[n]$ and all instances of the algorithm $A_{i,r}$ with $i\in[\log n]$, $r\in[R]$, and $R=\Theta\left(\frac{\log^{10}n}{\eps^5}\right)$, assuming $\eps=\frac{1}{\poly(n)}$.

\section{Intuition on \lemref{lem:hh:to:out}}
\applab{app:estimate}
The main intuition of \lemref{lem:hh:to:out} is to decompose a symmetric norm $\ell(x)$ on input vector $x$ into the contribution by each of its coordinates. 
The coordinates can then be partitioned into level sets, based on how much they contribute to the norm $\ell(x)$. 
The celebrated Indyk-Woodruff norm estimation sketch~\cite{IndykW05,BravermanOR15,WoodruffZ18,WoodruffZ21} can then be applied to approximate each of the level sets, by subsampling the universe and estimating the sizes of each universe through the heavy-hitters of each subsample. 

\begin{definition}[Important Levels]
For $x\in\mathbb{R}^n$ and $\alpha>0$, we define the level $i$ as the set $B_i=\{j\in[n]\,:\,\alpha^{i-1}\le|x_j|\le\alpha^i\}$. 
We use $b_i:=|B_i|$ to denote the size of level $i$. 
Then level $i$ is \emph{$\beta$-important} if
\[b_i>\beta\sum_{j>i}b_j,\qquad b_i\alpha^{2i}\ge\beta\sum_{j\le i}b_j\alpha^{2j}.\]
\end{definition}
Intuitively, a level is important if its size is significant compared to all the higher levels and its contribution is significant compared to all the lower levels. 
We shall show that identifying the important levels and their sizes for a certain base $\alpha$ and parameter $\beta$ suffices to approximate a symmetric norm $\ell(x)$. 
\begin{definition}[Level Vectors and Buckets]
Given a vector $x\in\mathbb{R}^n$ and the notation for the levels of $x$, the \emph{level vector} for $x$ is
\begin{align*}
V(x):=(\underbrace{\alpha^1,\ldots,\alpha^1}_{b_1\text{ times}},\underbrace{\alpha^2,\ldots,\alpha^2}_{b_2\text{ times}},\ldots,\underbrace{\alpha^k,\ldots,\alpha^k}_{b_k\text{ times}},0,\ldots,0)\in\mathbb{R}^n.
\end{align*}
%\begin{align*}
%V(x):=&(\underbrace{\alpha^1,\ldots,\alpha^1}_{b_1\text{ times}},\underbrace{\alpha^2,\ldots,\alpha^2}_{b_2\text{ times}},\ldots,\\
%&\underbrace{\alpha^k,\ldots,\alpha^k}_{b_k\text{ times}},0,\ldots,0)\in\mathbb{R}^n.
%\end{align*}
The $i$-th bucket of $V(x)$ is
\begin{align*}
V_i(x):=(\underbrace{0,\ldots,0,}_{b_1+\ldots+b_{i-1}\text{ times}}\underbrace{\alpha^i,\ldots,\alpha^i}_{b_i\text{ times}},\ldots,\underbrace{0,\ldots,0}_{b_{i+1}+\ldots+b_k\text{ times}},0,\ldots,0)\in\mathbb{R}^n.
\end{align*}
%\begin{align*}
%V_i(x):=&(\underbrace{0,\ldots,0,}_{b_1+\ldots+b_{i-1}\text{ times}}\underbrace{\alpha^i,\ldots,\alpha^i}_{b_i\text{ times}},\ldots,\\
%&\underbrace{0,\ldots,0}_{b_{i+1}+\ldots+b_k\text{ times}},0,\ldots,0)\in\mathbb{R}^n.
%\end{align*}
The values $\widehat{V(x)}$ and $\widehat{V_i(x)}$ given approximations $\widehat{b_1},\ldots,\widehat{b_k}$ for $b_1,\ldots,b_k$ are defined similarly. 
We use $V(x)\setminus V_i(x)$ to denote the vector that replaces the $i$-th bucket in $V(x)$ with all zeros. 
Similarly, $V(x)\setminus V_i(x)\cup\widehat{V_i(x)}$ replaces the $i$-th bucket in $V(x)$ with $\widehat{b_i}$ instances of $\alpha^i$. 
We omit the input $x$ when the dependency is clear from context. 
\end{definition}
We first relate approximating a symmetric norm to the concept of contributing levels and we will ultimately show the relationship between contributing levels and important levels. 
\begin{definition}[Contributing Levels]
Level $i$ of $x\in\mathbb{R}^n$ is \emph{$\beta$-contributing} if $\ell(V_i(x))\ge\beta\ell(V(x))$.  
\end{definition}
The following lemma states that a good approximation to $\ell(V)$ can be obtained even if all levels that are not $\beta$-contributing are removed. 
\begin{lemma}\cite{BlasiokBCKY17}
Let $V'$ be the vector obtained by removing all levels that are not $\beta$-contributing from $V$. 
Then $(1-\O{\log_{\alpha}n}\cdot\beta)\ell(V)\le\ell(V')\le\ell(V)$. 
\end{lemma}
Thus for sufficiently small $\beta$, approximating the symmetric norm $\ell(V)$ reduces to identifying the $\beta$-contributing levels:
\begin{lemma}
\lemlab{lem:reconstruction}
\cite{BlasiokBCKY17}
For precision $\eps>0$, let base $\alpha=(1+\O{\eps})$, importance parameter $\beta=\O{\frac{\eps^5}{\mmc(\ell)^2\cdot\log^5(n)}}$, and $\eps'=\O{\frac{\eps^2}{\log n}}$. 
Let $\widehat{b_i}\le b_i$ for all $i$ and $\widehat{b_i}\ge(1-\eps')b_i$ for all $\beta$-important levels. 
Let $\widehat{V}$ be the level vector constructed using $\alpha,\widehat{b_1},\ldots$ and $V'$ be the vector constructed by removing all the buckets that are not $\beta$-contributing in $\widehat{V}$.  
Then $(1-\eps)\ell(x)\le\ell(V')\le\ell(x)$. 
\end{lemma}
The following pair of lemmas provide intuition on how to identify $\beta$-contributing levels. 
\begin{lemma}
\lemlab{lem:contribute:num}
\cite{BlasiokBCKY17}
If level $i$ is $\beta$-contributing, then there exists some fixed constant $\lambda>0$ such that 
\[b_i\ge\frac{\lambda\beta^2}{\mmc(\ell)^2\log^2 n}\cdot\sum_{j>i}b_j.\] 
\end{lemma}

\begin{lemma}
\lemlab{lem:contribute:weight}
\cite{BlasiokBCKY17}
If level $i$ is $\beta$-contributing, then there exists some fixed constant $\lambda>0$ such that 
\[b_i\alpha^{2i}\ge\frac{\lambda\beta^2}{\mmc(\ell)^2(\log_{\alpha} n)\log^2 n}\cdot\sum_{j\le i}b_j\alpha^{2j}.\] 
\end{lemma}
Namely, \lemref{lem:contribute:num} and \lemref{lem:contribute:weight} imply that a level $i$ that is $\beta$-contributing must be an important level. 
Moreover, the problem of approximating the size of each important level can be reduced to the task of finding the $\nu$-approximate $\eta$-heavy hitters. 
\begin{lemma}
\lemlab{lem:iw:sketch}
\cite{BlasiokBCKY17}
For level base $\alpha>0$, importance parameter $\beta>0$ and precision $\eps'>0$, there exist parameters $\eta,\nu>0$ as defined in \lemref{lem:hh:to:out}, such that a $\nu$-approximate $\eta$-heavy hitters algorithm can be used to output a $(1+\eps')$-approximation to the size $b_i$ of all $\beta$-important levels.  
\end{lemma}
The subroutine $\estimate$ of \lemref{lem:hh:to:out} reconstructs an estimate of the level vector by removing all the levels that are not $\beta$-contributing and using a $(1+\eps')$-approximation to the sizes of all $\beta$-important levels. 
It follows by \lemref{lem:reconstruction} that this procedure suffices to obtain a $(1+\eps)$-approximation to $\ell(x)$, thus (informally) justifying the correctness of $\estimate$. 

The $(1+\eps')$-approximation to the size $b_i$ of all $\beta$-important levels guaranteed by \lemref{lem:iw:sketch} is not immediate from the $\nu$-approximate $\eta$-heavy hitters algorithm. 
Rather, the algorithm to approximately recover the size $b_i$ of all $\beta$-important levels uses the same intuition as the Indyk-Woodruff sketch~\cite{IndykW05}. 
The main observation is that each $\beta$-important level must have either large size or large contribution (or both). 
If the $\beta$-important level has large contribution but small size, then its elements will immediately be recognized as a heavy-hitter. 
Otherwise, if the $\beta$-important level has large size, then a large number of these coordinates will be subsampled and ultimately become heavy-hitters at some level $i$ in which $\Theta\left(\frac{1}{\eps^2}\right)$ of these coordinates are subsampled. 
The size $b_i$ can then approximately recovered by rescaling by the sampling probability, though additional care must be used to formalize this argument, e.g., by randomizing the boundaries of the level sets.

\section{Empirical Evaluations}
\applab{app:exp}
In this section, we evaluate the performance of our algorithm on both synthetic and real-world dataset.

\paragraph{Synthetic data.} 
We construct a synthetic stream as follows. 
We first generate an ordered list $s_1$ of $\frac{m}{4}$ numbers starting from $2$, i.e., $s_1=\left\{2, 3, 4,\ldots\frac{m}{4} + 1\right\}$. 
We then generate a random stream $s_2$ of size $\frac{m}{2}-\frac{m}{1000}$ from a universe of $n - \frac{m}{2} -1$ by drawing each element uniformly at random. 
That is, $x\sim U\left(\frac{m}{2} + 2, \frac{m}{2} + 3, ..., n\right)$ for each $x\in s_2$. 
We combine these three streams $S = s_1\circ s_1\circ s_2$, where $\circ$ denotes the concatenation of the streams. 
Finally, we fix the last $\frac{m}{1000}$ fraction of the stream to be $1$. 
Thus, we have a stream $S$ of length $m$ on a universe of size $n$ and for sufficiently large $m$, the stream contains a single $L_2$ heavy hitter (the element $1$). 
We run experiments on $m\in\{2^{10},2^{11},2^{12},2^{13},2^{14},2^{15}\}$. 
Moreover, we run experiments on both $W=m$ so that the window consists of the entire stream and $W=\frac{m}{2}$ so that the active elements are the latter half of the stream. 

\paragraph{CAIDA Anonymized Internet Traces 2019 Dataset.} 
For real-world data we use the ``Equinix-nyc-2019'' dataset from the Center for Applied Internet Data Analysis (CAIDA), which is collected by a monitor in New York City that is connected to an OC192 backbone link (9953 Mbps) of a Tier 1 Internet Service Provider (ISP) between New York, NY and Sao Paulo, Brazil. 
The infrastructure consists of 2 physical machines that each have a single Endace 6.2 DAG network monitoring card that is connected to a single direction of the bi-directional backbone link. 
The source IP addresses \texttt{src} are used as the input to our experiments.

\paragraph{Implementation.} 
All algorithms are implemented in Python 3.8.3 and are carried out on Intel Xeon Gold 6226 CPU and Tesla V100 16GB GPU. 
We test our algorithm on various normalization functions such as $L_p$ or top-$k$. 
We focus on the relative error of each algorithm, comparing the performance of our algorithm to both uniform sampling the stream with $0.1$ sampling rate and uniform sampling the universe with $0.1$ sampling rate when possible. 
The results are averaged over the number of rows in the sketch and we do not consider the time performance of the algorithms, which we consider beyond the scope of our paper. 

\paragraph{Results.} 
In \figref{fig:synth} and \figref{fig:synth:half}, we show how the various norm estimation errors perform on our synthetic dataset as the stream length changes, both for $W=m$ and $W=\frac{m}{2}$, across $m\in\{2^{10},2^{11},2^{12},2^{13},2^{14},2^{15}\}$. 
We observe that our algorithm consistently performs the best and although uniformly sampling from the universe performs poorly, uniformly sampling from the stream performs surprisingly well for smaller stream lengths. 
This is because with such a large sampling rate, the uniform sampling algorithms essentially use linear space. 
Nevertheless, our algorithm demonstrates superior performance compared with these baselines. 

In \figref{fig:caida}, we compare the various norm estimation errors perform on the CAIDA dataset for $W=m$, across $m\in\{2^{10},2^{11},2^{12},2^{13},2^{14},2^{15}\}$. 
Because the universe consists of all possible source IP addresses, it is not feasible to perform uniform sampling from the universe for the CAIDA dataset. 
However, our algorithm again exhibits superior performance compared with uniform sampling from the stream. 

\begin{figure*}[!htb]
\centering
    \begin{subfigure}[t]{0.3\textwidth}
		\centering
		\includegraphics[width = 0.8\textwidth]{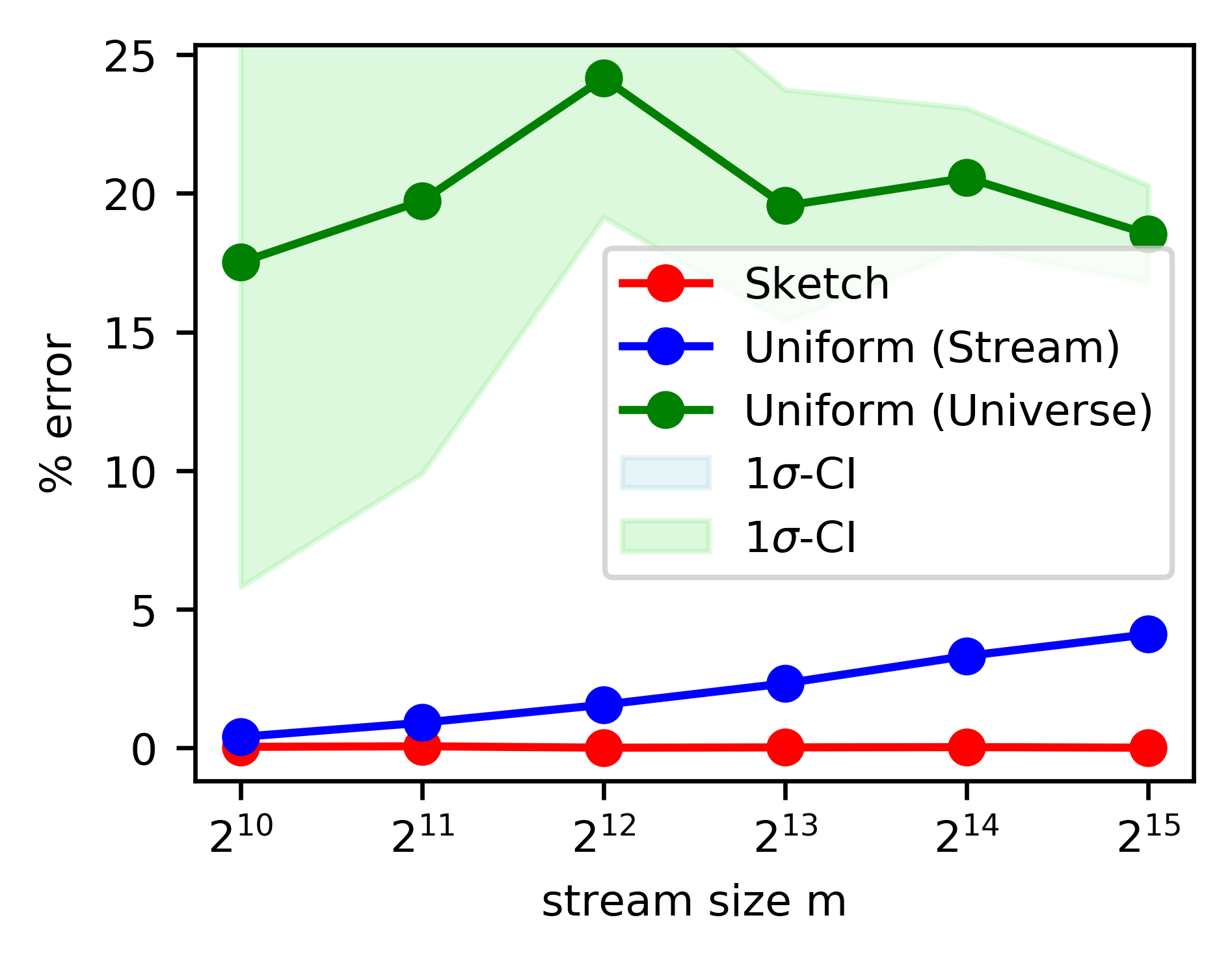}
    \caption{Synthetic data, $W=m$}
    \figlab{fig:synth}
    \end{subfigure}
    \begin{subfigure}[t]{0.3\textwidth}
		\centering
		\includegraphics[width = 0.8\textwidth]{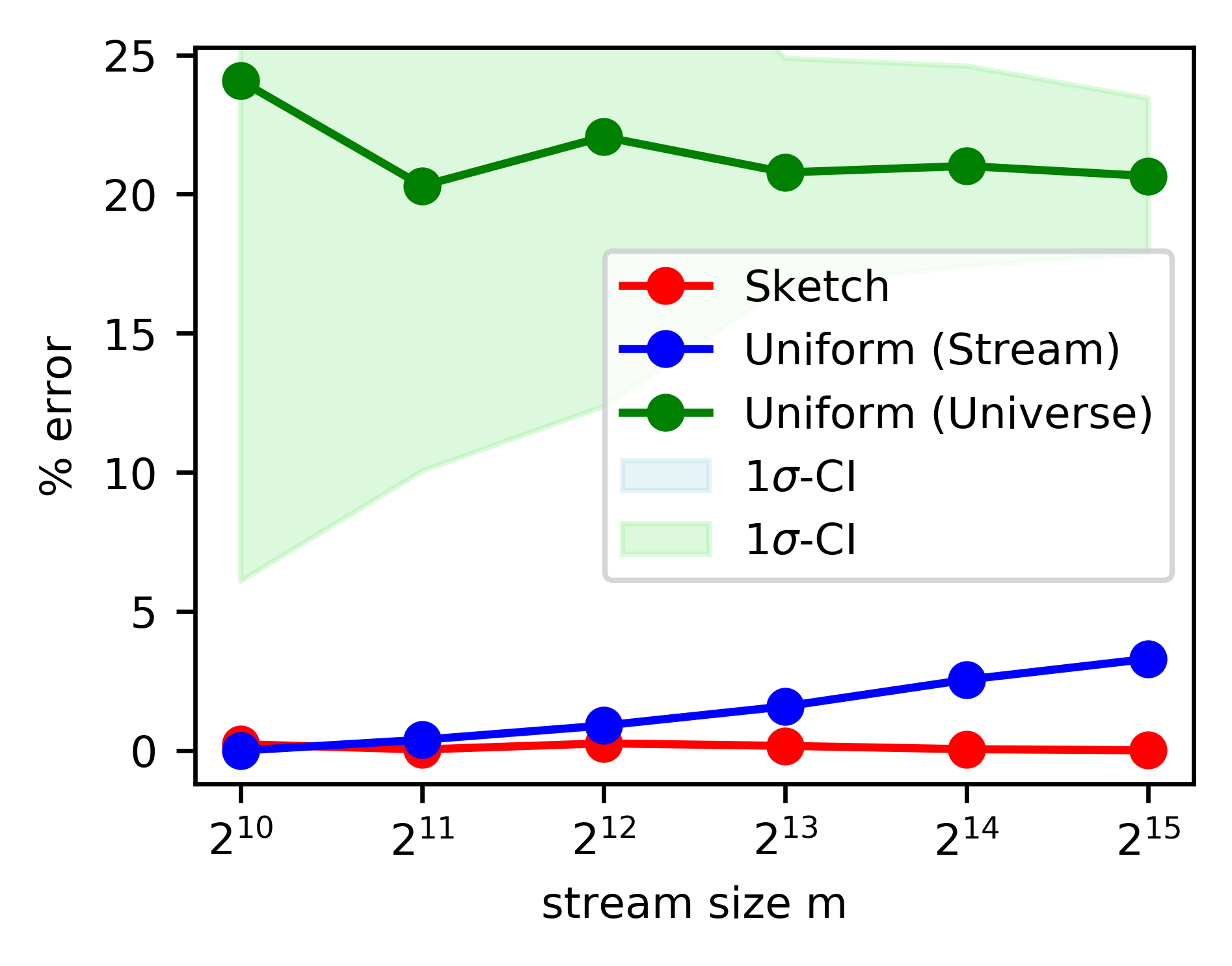}
    \caption{Synthetic data, $W=m/2$}
    \figlab{fig:synth:half}
    \end{subfigure}
		\begin{subfigure}[t]{0.3\textwidth}
		\centering
		\includegraphics[width = 0.8\textwidth]{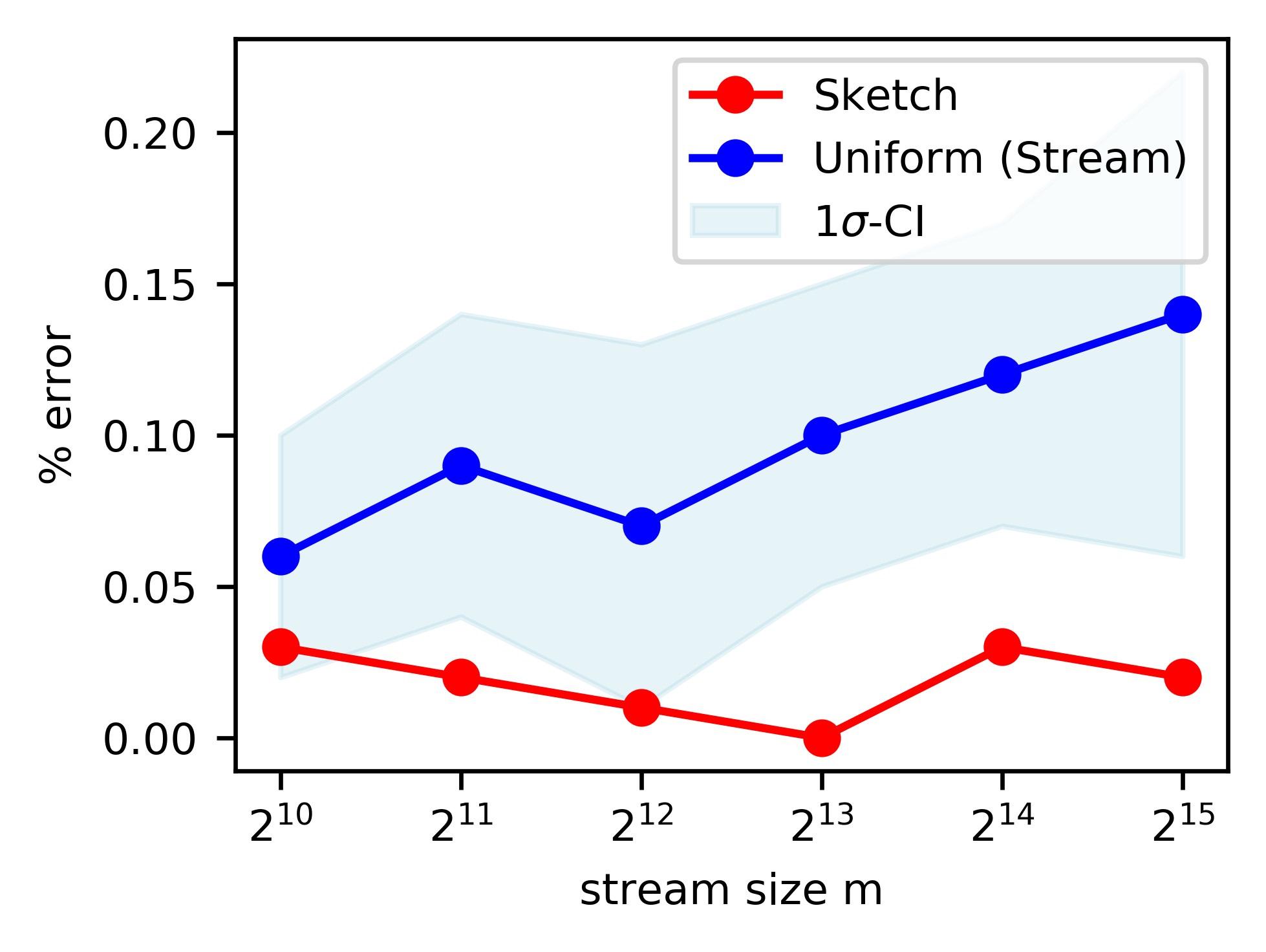}
    \caption{CAIDA}
		\figlab{fig:caida}
    \end{subfigure}
\caption{Relative error as a function of stream length}
\figlab{fig:exp}
\end{figure*}

\end{document}